\documentclass[11pt]{article}

\usepackage{fullpage}
\usepackage{soul}
\usepackage[utf8]{inputenc}
\usepackage[small]{caption}
\usepackage{xcolor}
\usepackage{amsmath}
\usepackage{booktabs}
\usepackage{framed}
\usepackage{tikz}
\usepackage{url}

\usepackage{amsthm}
\usepackage{amssymb}
\usepackage{afterpage}

\usepackage{algpseudocode,graphicx,caption,subcaption,xspace,nicefrac,mathtools,bm}

\definecolor{myred}{cmyk}{0, 0.75, .9, 0}
\definecolor{myblue}{cmyk}{.95, 0.45, 0, 0}
\definecolor{mygreen}{cmyk}{1, 0.15, 1, 0}
\definecolor{mygray}{gray}{0.3}
\usepackage{tikz}
\usetikzlibrary{arrows.meta,calc,shadows.blur}
\tikzstyle{ddotted}=[dash pattern=on 1pt off .3pt]
\tikzset{
  dot/.style = {circle, draw, fill, minimum size=#1,inner sep=0pt, outer sep=0pt, very thin},
  dot/.default = 4pt,
  path/.style = {thick, rounded corners=1pt, mygreen, -latex},
  overpath/.style = {green,line width=.1pt,rounded corners=1pt,shorten >=4pt}
}

\usepackage[ruled,vlined,linesnumbered,norelsize,noend]{algorithm2e}

\algnotext{EndFor}
\algnotext{EndIf}
\algnotext{EndWhile}
\algnotext{EndProcedure}

\newtheorem{theorem}{Theorem}[section]

\newtheorem{proposition}[theorem]{Proposition}
\newtheorem{lemma}[theorem]{Lemma}

\theoremstyle{definition}
\newtheorem{definition}[theorem]{Definition}
\newtheorem{example}[theorem]{Example}

\newcommand*{\innerproofname}{Proof}
\newenvironment{innerproof}[1][\innerproofname]{\begin{proof}[#1]}{\end{proof}}

\usepackage{natbib}
\usepackage{cleveref}
\usepackage{mathtools}
\usepackage{enumerate}

\newcommand{\sd}{\text{SD}}
\renewcommand{\emptyset}{\varnothing}

\allowdisplaybreaks

\usepackage{authblk}

\title{\bf Two-Sided Fairness in Many-to-One Matching}

\author[1]{Ayumi Igarashi} 
\author[2]{Naoyuki Kamiyama} 
\author[1]{Yasushi Kawase}
\author[3]{Warut Suksompong}
\author[4]{\authorcr Hanna Sumita}
\author[4]{Yu Yokoi}

\affil[1]{University of Tokyo, Japan}
\affil[2]{Kyushu University, Japan}
\affil[3]{National University of Singapore, Singapore}
\affil[4]{Institute of Science Tokyo, Japan}

\date{\vspace{-1cm}}

\begin{document}

\maketitle

\begin{abstract}
We consider a classic many-to-one matching setting, where participants need to be assigned to teams based on the preferences of both sides.
Unlike most of the matching literature, we aim to provide fairness not only to participants, but also to teams using concepts from the literature of fair division.
We present a polynomial-time algorithm that computes an allocation satisfying team-justified envy-freeness up to one participant, participant-justified envy-freeness, balancedness, Pareto optimality, and group-strategyproofness for participants, even in the possible presence of ties.
Our algorithm generalizes both the Gale--Shapley algorithm from two-sided matching as well as the round-robin algorithm from fair division.
We also discuss how our algorithm can be extended to accommodate quotas and incomplete preferences.
\end{abstract}

\section{Introduction}

Two-sided matching is a fundamental economic problem that involves finding a desirable match based on the preferences of both sides \citep{GusfieldIr89,Manlove13}.
Several applications of this problem are \emph{many-to-one} in nature, such as allocating medical residents to hospitals, students to colleges, and workers to firms---we shall refer to the two sides as \emph{participants} and \emph{teams}. 
The theory developed on two-sided matching has been applied to numerous real-world scenarios over the years, perhaps most notably to the allocation of medical students in the National Resident Matching Program of the United States \citep{RothSo90}.

In their seminal work, \citet{GaleSh62} proposed a notion of stability for two-sided matching, which is based on the concept of \emph{justified envy-freeness}:
a participant $p$ is said to have \emph{justified envy} toward another participant $q$ if $p$ prefers $q$'s team to $p$'s own team and, moreover, $q$'s team prefers $p$ to $q$.
While a justified envy-free allocation ensures fairness among participants, it may be highly unfair for certain teams, as the following example illustrates.
\begin{example}
\label{ex:indifference}
Consider a two-sided many-to-one matching instance where all participants are indifferent among all teams,\footnote{\citet[p.~269]{ErdilEr17} discussed the importance of allowing and handling ties in the preferences.
For example, when \citet{deHaanGaOo23} asked parents in Amsterdam to assign numerical values to schools for their children, some assigned equal values to certain schools even though they had the option of specifying finer preferences.
Similarly, the Scottish Foundation Allocation Scheme allows hospitals to express ties in their preferences over doctors, and most hospitals do express ties, with some ranking doctors in only three indifference classes.} and all teams share the same preference over participants.
Due to the (lack of) preference of the participants, every allocation is justified envy-free in this instance.
However, an allocation that assigns the strongest participants (according to the preference shared by all teams) to one team and the weakest participants to another team is likely to be perceived as unfair by the latter team, even though justified envy-freeness is fulfilled.\footnote{Moreover, the stability notion of \citet{GaleSh62} is fulfilled, assuming that every team has already filled its quota in the current allocation.}
Indeed, a hospital may well feel harshly treated if it receives only the worst doctors even though all doctors are indifferent among all hospitals; the same statement can be made for schools and students.    
For the sake of fairness among teams, it is therefore preferable to assign a set of participants of roughly equal desirability to every team in this instance.
A similar example can be constructed such that each participant has more than one indifference class in her preference.
\end{example}

Recently, \citet{IgarashiKaSu24} addressed the issue of team fairness using concepts from the literature of fair division \citep{BramsTa96,Moulin03,Moulin19}.
Canonical fair division can be viewed as a special case of two-sided matching where only one side (i.e., teams) has preferences, while the other side (i.e., participants) is completely indifferent.
A prominent fairness criterion in that literature is \emph{envy-freeness} (among teams), which stipulates that no team should envy another team based on the sets of participants that they receive \citep{Foley67,Varian74}.
Note that this criterion precisely addresses the unfairness that arises in \Cref{ex:indifference}.
However, envy-freeness cannot always be satisfied even by itself---consider, for example, when there is one exceptional participant that every team values more than all remaining participants combined.
In light of this, \citet{Budish11} proposed relaxing envy-freeness to \emph{envy-freeness up to one participant (EF1)}.
An allocation is said to be EF1 for teams if any envy that a team has toward another team can be eliminated by removing some participant from the latter team.
As we consider envy-freeness on both sides, for clarity, we shall refer to the notions  as \emph{team-EF1} and \emph{participant-justified EF}.
Assuming that teams have additive valuations over participants, a team-EF1 allocation always exists and can be found efficiently, for instance, using the \emph{round-robin algorithm}, where teams take turns picking their favorite participant from the remaining participants.\footnote{The round-robin algorithm is sometimes referred to as the \emph{draft} process when used in the context of sports.
}

In one-to-one matching, every allocation that assigns at most one participant to each team is team-EF1, so the existence of a team-EF1 and participant-justified EF allocation follows from the result of \citet{GaleSh62}.
Perhaps surprisingly, however, \citet{IgarashiKaSu24} gave an example demonstrating that such an allocation may not exist in the many-to-one setting.
\begin{example}[\citep{IgarashiKaSu24}]
\label{ex:incompatibility}
Assume that there are $n = 2$ teams and $m = 4$ participants.
Team~$1$ has value $3,3,2,2$ for the participants $p_1,p_2,p_3,p_4$, respectively, while team~$2$ has value $1,1,0,0$ for them.
The valuation of each team is additive: a team's value for a set of participants is simply the sum of its values for the participants in the set.
All four participants prefer team~$1$ to team~$2$.

In this example, team~$2$ needs at least one of $p_1,p_2$ in order for team-EF1 to be fulfilled---assume without loss of generality that it receives $p_2$.
Given this, team~$1$ needs at least one of $p_3,p_4$ for team-EF1 to be satisfied---assume without loss of generality that it receives $p_3$.
However, this results in $p_2$ having participant-justified envy toward~$p_3$.
\end{example}
\citet{IgarashiKaSu24} showed that this incompatibility persists even if we were to relax team-EF1 to ``team-EF$k$'' for any constant~$k$.
Nevertheless, consider the allocation that assigns all four participants to team~$1$.
Even though team~$2$ envies team~$1$ by more than one participant, this envy could be considered ``unjustified'', as all participants prefer team~$1$ to team~$2$.
Our first, conceptual, contribution is to introduce a relaxation of team-EF1 that we call \emph{team-justified EF1}.
An allocation is said to be team-justified EF1 if the following holds: whenever a team~$i$ compares its set of participants to that of another team~$j$, after removing the participants who strictly prefer $j$ to $i$ from the latter set, any envy that $i$ still has toward $j$ can be eliminated by removing one further participant from $j$'s set.
Observe that when all participants are completely indifferent, any envy by a team is justified, so team-justified EF1 reduces to team-EF1.
On the other hand, when participants have preferences over teams, team-justified EF1 allows us to take participant preferences into account in a meaningful manner.
Can we ensure fairness for both sides by guaranteeing team-justified EF1 and participant-justified EF, perhaps together with other desirable properties?

\subsection{Overview of Results}

Following the majority of the matching literature, we assume that each team (resp., participant) has an ordinal ranking over the participants (resp., teams), possibly with ties.
Under ordinal preferences, team-EF1 and team-justified EF1 can be defined in a natural manner using the \emph{stochastic dominance (SD)} relation.
We refer to \Cref{sec:prelims} for the formal definitions, but remark here that the resulting notions---\emph{team-SD-EF1} and \emph{team-justified SD-EF1}---imply their non-SD counterparts for any additive valuations consistent with the rankings.\footnote{Although the SD relation is often used to compare randomized outcomes, it is also natural for comparing deterministic ones. In particular, SD-EF1 has been studied in a number of papers \citep{FreemanMiSh21,AzizFrSh24,PrakashNi24}. \label{footnote:SD}}

In \Cref{sec:warmup}, we warm up by describing two approaches that satisfy some but not all of the desired properties.
The first approach is to run the round-robin algorithm but only allow each team to pick participants who rank the team first.
While this approach yields fairness to both sides in the form of team-justified SD-EF1 and participant-EF,\footnote{This is a strengthening of participant-justified EF where a participant's envy does not need to be justified.} it may result in an extremely unbalanced allocation.
The second approach involves creating an auxiliary instance, which is a one-to-one matching instance with strict preferences, by making copies that represent the slots in each team, and then computing any participant-justified EF matching in this instance (e.g., using the Gale--Shapley algorithm).
This approach guarantees team-justified SD-EF1, participant-justified EF, and balancedness.
However, it does not ensure \emph{swap stability}, i.e., there may exist a pair of participants such that swapping them makes some involved participant or team better off and none of them worse off.

In \Cref{sec:main-algo}, we present our main algorithm, which produces an allocation that fulfills all of the key properties we consider.
Specifically, the returned allocation always satisfies team-justified SD-EF1, participant-justified EF, \emph{balancedness}---which means that the numbers of participants allocated to any pair of teams differ by at most one---and \emph{Pareto optimality (PO)}, an efficiency notion that strengthens swap stability.
The algorithm, which can be implemented in polynomial time, creates team-slots for each team and maintains a set of eligible teams for each participant; initially, no team is eligible for any participant.
In each iteration, the algorithm computes an eligible matching between team-slots and participants that maximizes the team-slots' values in round-robin order, and breaks any ties by optimizing the participants' preferences lexicographically.
As long as the matching is incomplete, each unmatched participant's eligibility set is expanded to include the participant's most-preferred team(s) that are still ineligible.
When all participants are indifferent between all teams, our algorithm reduces to the round-robin algorithm, whereas when all teams and participants have strict preferences, it reduces to the Gale--Shapley algorithm.
In \Cref{sec:SP}, we demonstrate a further resemblance between our algorithm and Gale--Shapley by showing via an elaborate proof that our algorithm is \emph{group-strategyproof} for participants, that is, no group of participants can misreport their preferences in such a way that all of them strictly benefit.\footnote{
One could define an even stronger version of group-strategyproofness, where no group of participants can collude to misreport their preferences in a way that makes none of the members worse off and at least one member strictly better off.
However, even in one-to-one matching, this stronger version is violated by the Gale--Shapley algorithm \citep{Huang06}.}

Finally, in \Cref{sec:discussion}, we discuss how our algorithm can be extended to accommodate teams that have quotas as well as participants or teams that may prefer being unassigned.
In addition, we discuss our results in the context of stable matching with indifferences, by showing that they yield a mechanism that is simultaneously stable, PO, and group-strategyproof for participants.

\subsection{Further Related Work}

Our work combines elements of fair division \citep{BramsTa96,RobertsonWe98,Moulin03,Moulin19} as well as two-sided matching \citep{GusfieldIr89,RothSo92,Manlove13}.
While each of these two areas has an extensive literature on its own, few connections have been established between them thus far.

The vast majority of work in fair division assumes one-sided preferences---in our terminology, teams have preferences over participants, but participants do not have preferences over teams.
\citet{IgarashiKaSu24} studied fair division under two-sided preferences assuming that teams have cardinal utilities.
As mentioned earlier, they showed that an allocation satisfying team-EF1 and participant-justified EF may not exist; in fact, they also proved that deciding whether such an allocation exists is NP-hard.
Nevertheless, they devised a polynomial-time algorithm that computes an allocation satisfying team-EF1, balancedness, and swap stability.
\citet{JainVa24} investigated the complexity of maximizing the Nash welfare---that is, the geometric mean of the utilities---in this setting.
\citet{FreemanMiSh21} considered many-to-many matching and proposed the notion of ``double-EF1'', which requires EF1 to hold for both sides.
However, in many-to-one matching, EF1 is meaningless on the participant side, as it is always trivially satisfied.

In the two-sided matching literature, fairness is often considered in the form of participant-justified EF, typically referred to simply as justified EF \citep{AbdulkadirogluSo03,FragiadakisIwTr15,KamadaKo17,WuRo18,Yokoi20}.
As we discussed, this may lead to large (justified) envy among teams.
In addition, a significant portion of the matching literature is based on the assumption of strict preferences.
Specifically, \citet{DubinsFr81} and \citet{Roth82} proved that the (participant-proposing) Gale--Shapley algorithm is strategyproof for participants under the assumption of strict preferences.
However, as \citet[p.~269]{ErdilEr17} argued, ties in preferences are in fact widespread in practical applications, and the way these ties are handled can have important consequences.
Erdil and Ergin designed polynomial-time algorithms that compute a PO and ``stable'' matching as well as a participant-optimal stable matching; their notion of stability is similar to participant-justified EF but takes into account the ability of participants to switch to an empty slot in another team.
\citet{NarangBiNa22} examined the problem of finding a stable matching that is optimal for both sides with respect to the ``leximin'' ordering.

\section{Preliminaries}
\label{sec:prelims}

Let $T = [n]$ be the set of teams, where $[z] \coloneqq \{1,\dots,z\}$ for each positive integer~$z$, and let $P = \{p_1,\dots,p_m\}$ be the set of participants.
We sometimes refer to either a team or a participant as a \emph{party}.
Each team $i\in T$ has a complete and transitive preference $\succsim_i$ over the participants (allowing indifferences), with $\succ_i$ and $\sim_i$ denoting the strict and equivalence part of $\succsim_i$, respectively.
Similarly, each participant $p\in P$ has a weak transitive preference $\succsim_p$ over the teams.
For ease of reading, we sometimes write, e.g., $(p_1, p_2) \succ_i (p_3, p_4)$ instead of $p_1\sim_i p_2\succ_i p_3\sim_i p_4$.
An \emph{instance} consists of the set of teams $T$, the set of participants $P$, and the preferences of both sides.

An \emph{allocation} $A = (A_1,\dots,A_n)$ is an ordered partition of $P$ into $n$ parts, where the part $A_i$ is allocated to team~$i$.
We shall investigate several desirable properties of allocations, starting with fairness on the participant side.

\begin{definition}
Given an allocation~$A$, a participant $p\in A_i$ is said to have \emph{envy} toward another participant $p'\in A_j$ if $j\succ_p i$.
The envy is \emph{justified} if it additionally holds that $p\succ_j p'$.
An allocation is \emph{participant-EF} (resp., \emph{participant-justified EF}) if no participant has envy (resp., justified envy) toward another participant according to the allocation.
\end{definition}

Participant-justified EF forms part of the stability notion commonly studied in the matching literature (e.g., \citep{GaleSh62}), with the other part being ``non-wastefulness'', that is, no participant prefers to switch to a team with sufficient capacity to accept the participant.
However, we will only consider team capacities in \Cref{sec:discussion}.

Next, we consider fairness on the team side.
Given a team $i\in T$, the \emph{stochastic dominance (SD)} relation $\succsim_i^{\sd}$ is a partial order of sets of participants defined as follows (see also \Cref{footnote:SD}).
For any two sets of participants $Q, R\subseteq P$, it holds that $Q\succsim_i^{\sd} R$ if and only if there exist
\begin{enumerate}[(i)]
\item a subset $Q'\subseteq Q$ with $|Q'| = |R|$, and
\item a bijection $\pi\colon Q'\to R$ such that $p\succsim_i\pi(p)$ for every $p\in Q'$.
\end{enumerate}
For example, if a team~$i$ has the ranking $p_1\succ_i p_2\succ_i p_3\succ_i p_4$, then $\{p_1,p_3\}\succ_i^\sd \{p_2,p_4\} \succ_i^\sd \{p_3\}$, whereas $\{p_1,p_4\}\not\succsim_i^\sd \{p_2,p_3\}$ and $\{p_2,p_3\}\not\succsim_i^\sd \{p_1,p_4\}$.
Note that the relation $\succsim_i^\sd$ is transitive but not necessarily complete.

\begin{definition}
An allocation~$A$ is said to satisfy 
\begin{itemize}
\item \emph{team-SD-EF1} if for all distinct $i, j\in T$, it holds that $A_i\succsim_i^\sd (A_j\setminus X)$ for some $X\subseteq A_j$ with $|X|\le 1$;
\item \emph{team-justified SD-EF1} if for all distinct $i, j\in T$, it holds that $A_i\succsim_i^\sd (B_j\setminus X)$ for some $X\subseteq B_j$ with $|X|\le 1$, where $B_j$ denotes the set of participants in $A_j$ who weakly prefer team~$i$ to team~$j$.
\end{itemize}
\end{definition}

In any instance, a team SD-EF1---which must also be team-justified SD-EF1---exists and can be found by the \emph{round-robin algorithm}, which lets the teams pick their favorite participant in the order $1,2,\dots,n,1,2,\dots$, until all participants have been assigned.
If teams have additive valuations over participants and their values for individual participants are consistent with the rankings, then team-SD-EF1 (resp., team-justified SD-EF1) with respect to the rankings implies ``team-EF1'' (resp., ``team-justified EF1'') with respect to the additive valuations.
More generally, suppose that teams have \emph{responsive} preferences over sets of participants, meaning that if a team prefers participant $p$ to $p'$, then it also prefers the set $Q\cup\{p\}$ to $Q\cup\{p'\}$ for any $Q\subseteq P$ such that $p,p'\not\in Q$.
Then, team-SD-EF1 (resp., team-justified SD-EF1) with respect to the rankings over participants again implies ``team-EF1'' (resp., ``team-justified EF1'') with respect to the rankings over \emph{sets} of participants.\footnote{While one could consider an even more general setting where teams have ordinal preferences over sets of participants and these preferences are not necessarily responsive, even with one-sided preferences, it is unknown whether (team-)EF1 alone can always be satisfied \citep[Open problem 4.1]{Suksompong21}.}

Another property that can be useful in many-to-one matching is \emph{balancedness}, which requires the participants to be distributed as equally among the teams as possible.
Note that an allocation returned by the round-robin algorithm is guaranteed to be balanced.

\begin{definition}
An allocation $A$ is said to be \emph{balanced} if $\big||A_i|-|A_j|\big|\le 1$ for all $i,j\in T$.
Equivalently, an allocation $A$ is balanced if $|A_i|\in\{\lfloor m/n\rfloor, \lceil m/n\rceil\}$ for all $i\in T$.
\end{definition}

Besides fairness, one may also want the allocation to be optimal with respect to some notion of optimality.
We consider two natural optimality notions.

\begin{definition}
Given an allocation~$A$, a swap between participants $p\in A_i$ and $p'\in A_j$ (for some $i,j\in T$) is a \emph{beneficial swap} if it makes all four involved parties weakly better off and at least one party strictly better off (where we consider the SD relation for teams).
Another allocation $A'$ \emph{Pareto dominates} $A$ if all parties are weakly better off in $A'$ than in $A$, with at least one party being strictly better off; in this case, we say that $A'$ is a \emph{Pareto improvement} of~$A$.

An allocation $A$ is said to be \emph{swap stable} if it does not admit a beneficial swap, and \emph{Pareto optimal (PO)} if it does not admit a Pareto improvement.
\end{definition}

\section{Warm-up: Two Approaches}
\label{sec:warmup}

In this section, we describe two approaches that fulfill some but not all of the desired properties.

The first approach is to run the round-robin algorithm, but only allow each team to pick participants who rank the team first.
If a team has no such participant left, the team's turn is skipped.
This approach is formalized as \Cref{alg:top-choice}.

\begin{algorithm}
  \DontPrintSemicolon
  \SetAlgoLined
  \For{$i\in T$}{
    $P_i \leftarrow \{p\in P \mid i\succsim_p j \text{ for all } j\in T\}$\;
    $A_i \leftarrow \emptyset$
  }
  $R \leftarrow P$\;
  \While{$R\ne\emptyset$}{
    \For{$i \leftarrow 1$ \KwTo $n$}{ \label{line:for-loop}
        \If{$P_i\cap R\ne \emptyset$}{
            $p \leftarrow$ most preferred participant in $P_i\cap R$ according to $\succsim_i$ \;
            $A_i \leftarrow A_i\cup\{p\}$ and
            $R \leftarrow R\setminus\{p\}$\;
        }
    }
  }
    \Return $(A_1,\dots,A_n)$
    \caption{Round-robin over top-ranking participants}
    \label{alg:top-choice}
\end{algorithm}

\begin{theorem}
\Cref{alg:top-choice} returns an allocation that satisfies team-justified SD-EF1 and participant-EF.
\end{theorem}

\begin{proof}
Participant-EF is clear since each participant is assigned to her most preferred team.
For team-justified SD-EF1, consider two teams $i$ and $j$, and assume first that $i < j$.
In each iteration of the for-loop in \Cref{line:for-loop}, if $i$ likes the participant that it receives less than the participant that $j$ receives, that must be because the latter participant is in $P_j$ but not $P_i$, in which case the participant strictly prefers $j$ to $i$.
Hence, either $i$ does not envy $j$ with respect to this round, or $i$'s envy toward $j$ is not justified.
A similar argument applies when $i > j$ by considering the participant that $i$ receives in an iteration and the participant that $j$ receives in the next iteration, so if we remove the participant that $j$ receives in the first iteration, then $i$ does not have justified envy toward~$j$.
Thus, we can conclude that team-justified SD-EF1 from $i$ toward $j$ is satisfied.
\end{proof}

While \Cref{alg:top-choice} produces an allocation that is fair for both sides, the resulting allocation may be extremely unbalanced.
For example, if all participants have team~$1$ as their unique favorite team, then all of them will be assigned to team~$1$, which in many applications would be undesirable.
For similar reasons, participant-EF and balancedness cannot always be attained simultaneously, so a relaxation to participant-justified EF is necessary.

We next present an approach that yields team-justified SD-EF1 and participant-justified EF along with balancedness.
The approach involves creating an auxiliary instance, which is a one-to-one matching instance with strict preferences, where the team side consists of the slots within each team.
Each team-slot's preference over participants follows that of the original team (with ties broken arbitrarily).
Each participant's preference over team-slots also follows her preference over the original teams, but with ties broken in an ``interleaved'' manner.
For example, if a participant~$p$ has the preference $(1, 2)\succ_p (3, 4)$ over four teams and there are three slots in each team (denoted by subscripts), then the participant's preference in the auxiliary instance is 
\[
1_1 \succ_p 2_1 \succ_p 1_2 \succ_p 2_2 \succ_p 1_3 \succ_p 2_3 \succ_p 3_1 \succ_p 4_1 \succ_p 3_2 \succ_p 4_2 \succ_p 3_3 \succ_p 4_3.
\]
We then compute any participant-justified EF matching in the auxiliary instance (e.g., using the Gale--Shapley algorithm), and return the allocation where each team is assigned the participants that its slots receive.
This approach is formalized as \Cref{alg:auxiliary}.

\begin{algorithm}[t]
  \DontPrintSemicolon
  \SetAlgoLined
  $k \leftarrow \lfloor m/n\rfloor$ \;
  $r \leftarrow m - kn$ \;
  \lFor{$i \leftarrow 1$ \KwTo $r$}{
    Create $k+1$ slots $i_1, \dots, i_{k+1}$ of team~$i$.
  }
  \lFor{$i \leftarrow r+1$ \KwTo $n$}{
    Create $k$ slots $i_1, \dots, i_{k}$ of team~$i$.
  }
  The auxiliary instance is a one-to-one matching instance which consists of the slots of all teams (a total of $m$ slots) and the $m$ original participants. \;
  Each team-slot has a strict preference over participants which is the same as the original team's preference, with ties broken arbitrarily. \;
  Each participant has a strict preference over team-slots which follows the participant's preference over original teams, with ties broken in an ``interleaved'' manner.
  Specifically, if a participant strictly prefers team $i$ to team~$j$, then it prefers all slots of $i$ to all slots of~$j$.
  If the participant likes $i$ and $j$ equally, where $i < j$, then it prefers $i_\ell$ to $j_{\ell'}$ whenever $\ell \le \ell'$. \;
  $S \leftarrow $ any participant-justified EF matching in the auxiliary instance (e.g., computed using the Gale--Shapley algorithm) \;
  
    \Return the allocation where each team is assigned the participants that its slots receive in~$S$
    \caption{Auxiliary instance approach}
    \label{alg:auxiliary}
\end{algorithm}

\begin{theorem}
\label{thm:auxiliary}
\Cref{alg:auxiliary} returns an allocation that satisfies team-justified SD-EF1, participant-justified EF, and balancedness.
\end{theorem}

\begin{proof}
Balancedness follows immediately from the fact that each team receives either $\lfloor m/n\rfloor$ or $\lceil m/n\rceil$ participants from the algorithm.
For participant-justified EF, suppose for contradiction that a participant $p\in A_i$ envies another participant $p'\in A_j$ (i.e., $j\succ_p i$), and moreover the envy is justified (i.e., $p\succ_j p'$).
Then, in the auxiliary instance, $p$ prefers all slots of $j$ to all slots of $i$, and moreover all slots of $j$ prefer $p$ to $p'$.
Hence, no matter which team-slots the participants $p$ and $p'$ are assigned to in the auxiliary instance, participant-justified EF of the matching in the auxiliary instance is violated, a contradiction.

It remains to show team-justified SD-EF1.
Consider any two teams $i$ and $j$, and the participants $p$ and $p'$ assigned to the team-slots $i_\ell$ and $j_{\ell + 1}$, respectively, for some $\ell$.
If $p'$ strictly prefers $j$ to~$i$, then $p'$ does not count in the justified envy of $i$ toward $j$.
Else, $p'$ weakly prefers $i$ to $j$, so $p'$ strictly prefers $i_\ell$ to $j_{\ell + 1}$ in the auxiliary instance.
Since $p'$ is assigned to $j_{\ell+1}$ and $p$ is assigned to $i_\ell$, by participant-justified EF and the fact that all preferences are strict in the auxiliary instance, $i_\ell$ must strictly prefer $p$ to $p'$.
This means that $i$ weakly prefers $p$ to $p'$ in the original instance.
Hence, when we compare the slot~$i_\ell$ with the slot~$j_{\ell+1}$, either
\begin{enumerate}
\item[(a)] the participant assigned to $j_{\ell+1}$ strictly prefers $j$ to $i$, in which case this participant does not contribute to the justified envy from $i$ to $j$; or
\item[(b)] $i$ weakly prefers the participant assigned to~$i_\ell$ to the participant assigned to $j_{\ell+1}$.
\end{enumerate}

Assume that $i$ receives $k$ slots and $j$ receives $k+1$ slots (other cases can be handled similarly).
By the argument above, we have that according to the SD relation, $i$ prefers the set of participants assigned to $i_1,\dots,i_k$ (i.e., $i$'s bundle) at least as much as the set of participants assigned to $j_2,\dots,j_{k+1}$ who do not strictly prefer $j$ to~$i$ (i.e., the ``justified'' part of $j$'s bundle minus~$j_1$).
It follows that $i$ is team-justified SD-EF1 toward $j$, as desired.
\end{proof}

Can we achieve swap stability in addition to the properties provided by \Cref{thm:auxiliary}?
A natural attempt is to run \Cref{alg:auxiliary} and then allow participants to make beneficial swaps as long as the allocation admits at least one such swap.
Unfortunately, the following example shows that this approach may break team-justified SD-EF1.

\begin{example}
Consider an instance with $n = 3$ teams and $m = 6$ participants.
Teams $1$ and $2$ are indifferent between all participants, while team~$3$ has the preference 
$
(p_1, p_2)\succ (p_3, p_4, p_5, p_6)
$.
Participant~$p_5$ has the preference $1\succ (2, 3)$, while all other participants are indifferent between all teams.

In the auxiliary instance, assume that all team-slots have the preference $p_1\succ p_2\succ p_3\succ p_4\succ p_5\succ p_6$; note that this is a possible version of team~$3$'s preference with ties broken.
Participant~$p_5$ has the preference $1_1\succ 1_2\succ 2_1\succ 3_1\succ 2_2\succ 3_2$, while all other participants have the preference $1_1\succ 2_1\succ 3_1\succ 1_2\succ 2_2\succ 3_2$.

One can check that the unique participant-justified EF matching in this auxiliary instance assigns participants $p_1,\dots,p_6$ to slots $1_1,2_1,3_1,1_2,2_2,3_2$, respectively, leading to the allocation $(\{p_1,p_4\}, \{p_2, p_5\}, \{p_3, p_6\})$ in the original instance.
The swap between $p_1$ and $p_5$ is a beneficial swap, since it makes all involved parties weakly better off and $p_5$ strictly better off.
After making this swap, we arrive at the allocation $(\{p_4,p_5\}, \{p_1, p_2\}, \{p_3, p_6\})$.
However, in this allocation, team~$3$ is not team-justified SD-EF1 toward team~$2$.
\end{example}

As another idea for achieving swap stability, instead of breaking ties from the original instance arbitrarily in the auxiliary instance, we may try keeping these ties intact.
We then find a participant-justified EF matching in the auxiliary instance which is optimal among all such matchings according to some optimality notion---a natural choice is PO.
The hope is that this could allow us to avoid the suboptimality caused by the arbitrariness in the tie-breaking.
However, as the following example demonstrates, this approach also fails to guarantee swap stability.

\begin{example}
Consider an instance with $n = 2$ teams and $m = 4$ participants.
Team~$1$ has the preference $p_1\succ p_2\succ p_3\succ p_4$, while team~$2$ has the preference $(p_2, p_3)\succ (p_1, p_4)$.
Participant~$p_1$ prefers team~$1$ to team~$2$, participant~$p_4$ prefers team~$2$ to team~$1$, while participants $p_2$ and $p_3$ are indifferent between both teams.

In the auxiliary instance, assume that we keep the preference with ties for copies of team~$2$.
Similarly, let $p_1$ have the preference $(1_1, 1_2)\succ (2_1, 2_2)$, $p_4$ have the preference $(2_1, 2_2)\succ (1_1, 1_2)$, and $p_2$ and $p_3$ have the preference $(1_1, 2_1)\succ (1_2, 2_2)$.\footnote{The example still works even if we let $p_1$ have the preference $1_1\succ 1_2\succ 2_1\succ 2_2$ and $p_4$ have the preference $2_1\succ 2_2\succ 1_1\succ 1_2$.}

Consider the matching that assigns $p_1$ to team-slot~$1_1$, $p_2$ to $2_1$, $p_3$ to $1_2$, and $p_4$ to $2_2$.
One can check that this matching is PO within the set of participant-justified EF matchings (in fact, it is PO even within the set of all matchings).
This translates to an allocation in the original instance that assigns $p_1, p_3$ to team~$1$ and $p_2, p_4$ to team~$2$.
However, the swap between $p_2$ and $p_3$ is a beneficial swap, so the allocation is not swap stable.
\end{example}

\section{Main Algorithm}
\label{sec:main-algo}

In this section, we establish the existence of an allocation that simultaneously fulfills all of the key properties we consider: team-justified SD-EF1, participant-justified EF, balancedness, and PO.

The algorithm is shown as \Cref{alg:main}.
At a high level, like \Cref{alg:auxiliary}, it creates team-slots whose values for participants are consistent with those of the corresponding teams.
It also maintains an ``eligibility set'' of teams for each participant, which starts by being empty.
At each step, the algorithm computes an eligible matching between team-slots and participants that maximizes the team-slots' values in round-robin order, where an unmatched team-slot is considered to receive value $-\infty$.
If there is more than one such matching, the algorithm further breaks ties by optimizing the participants' preferences lexicographically, from $p_1$ to $p_m$.
As long as the matching is incomplete, the algorithm expands the eligibility set of each unmatched participant by including the participant's most preferred team(s) that are still ineligible at that point.

\begin{algorithm}[t]
  \DontPrintSemicolon
  \SetAlgoLined
  $k \leftarrow \lfloor m/n\rfloor$ \;
  $r \leftarrow m - kn$ \;
  \lFor{$i \leftarrow 1$ \KwTo $r$}{
    Create $k+1$ slots $i_1, \dots, i_{k+1}$ of team~$i$.
  }
  \lFor{$i \leftarrow r+1$ \KwTo $n$}{
    Create $k$ slots $i_1, \dots, i_{k}$ of team~$i$.
  }
  \lFor{$j \leftarrow 1$ \KwTo $m$}{
    $E_j \leftarrow \emptyset$ \qquad \emph{\small // eligibility set of participant $p_j$}
  }
  Let each team-slot have arbitrary (real number) values for participants consistent with the preference of the corresponding team. \;
  $S \leftarrow$ empty matching between team-slots and participants \;
  \While{$S$ does not fully match all team-slots and participants \label{line:augment-eligibility}} {
    \For{$j \leftarrow 1$ \KwTo $m$}{
      \If{$p_j$ is unmatched in $S$}{
        Add to $E_j$ the most preferred team(s) of participant $p_j$ not already in $E_j$.
        (If there is more than one such team, add all of them.)
      }
    }
    $S \leftarrow$ a matching between team-slots and participants, where each participant $p_j$ can only be matched to team-slots corresponding to teams in $E_j$.
    Choose such a matching that maximizes the team-slots' values in the order $1_1, 2_1, \dots, n_1, 1_2, 2_2, \dots$; an unmatched team-slot is considered to receive value $-\infty$.
    If there is more than one such matching, choose one that optimizes the participants' preferences in the order $p_1,\dots, p_m$. \label{line:lexicographic-optimal} \;
  }  
  
    \Return the allocation where each team is assigned the participants that its slots receive in~$S$
    \caption{Main algorithm}
    \label{alg:main}
\end{algorithm}

Before analyzing \Cref{alg:main}, we present an example to illustrate how it works.

\begin{example}
Consider an instance with $n = 3$ teams and $m = 6$ participants.
Team~$1$ has the preference $(p_3, p_4)\succ (p_1, p_2, p_5, p_6)$, while teams~$2$ and $3$ are indifferent between all participants.
Participants $p_1,p_2$ have the preference $(1, 2)\succ 3$, participants $p_3,p_4$ have the preference $2\succ 3 \succ 1$, and participants $p_5,p_6$ have the preference $2\succ 1\succ 3$.
We consider the iterations of the while-loop in \Cref{line:augment-eligibility}.
\begin{itemize}
\item In the first iteration, $E_1 = E_2 = \{1,2\}$ and $E_3 = E_4 = E_5 = E_6 = \{2\}$.
A possible matching~$S$ assigns $1_1$ to $p_1$, $1_2$ to $p_2$, $2_1$ to $p_3$, $2_2$ to $p_4$, leaving $3_1$, $3_2$, $p_5$, and $p_6$ unassigned.
\item In the second iteration, $E_5$ and $E_6$ are expanded to $\{1,2\}$.
The same matching $S$ as before can still be chosen.
\item In the third iteration, $E_5$ and $E_6$ are expanded to $\{1,2,3\}$.
A possible matching $S$ assigns $1_1$ to $p_1$, $1_2$ to $p_2$, $2_1$ to $p_3$, $2_2$ to $p_4$, $3_1$ to $p_5$, and $3_2$ to $p_6$.
Since this matching is complete, the while-loop terminates.
\end{itemize}

The returned allocation is $(\{p_1, p_2\}, \{p_3, p_4\}, \{p_5, p_6\})$.
Note that this allocation is team-justified SD-EF1 (in particular, the envy of team~$1$ toward team~$2$ is not justified), participant-justified EF (in particular, the envy of $p_5,p_6$ toward $p_1,p_2,p_3,p_4$ is not justified), balanced, and PO.
\end{example}

We now establish the claimed properties of \Cref{alg:main}.

\begin{theorem}
\label{thm:main-properties}
\Cref{alg:main} returns an allocation that satisfies team-justified SD-EF1, participant-justified EF, balancedness, and PO (and therefore swap stability).  
\end{theorem}

\begin{proof}
First, observe that whenever the matching $S$ is incomplete, the algorithm expands the eligibility set of at least one participant by adding at least one team.
When all teams are eligible for all participants, a valid allocation is returned.
Hence, the algorithm is well-defined.
Balancedness follows from the fact that each team receives either $\lfloor m/n\rfloor$ or $\lceil m/n\rceil$ participants from the algorithm, and swap stability is an immediate consequence of PO.
Hence, we may focus on team-justified SD-EF1, participant-justified EF, and PO.
We establish each of these properties in turn.

\vspace{2mm}

\emph{Team-justified SD-EF1:}
Similarly to the proof of \Cref{thm:auxiliary}, it suffices to show that each team-slot does not have justified envy toward another team-slot that comes later in the round-robin order.
Suppose for contradiction that a team-slot $x$ is assigned participant $p$ but has justified envy toward a later team-slot $x'$ which is assigned participant $p'$.
This means that $x$ values $p'$ more than $p$, and moreover $p'$ weakly prefers $x$ to $x'$.
Since $x'$ is eligible for $p'$, so is $x$.
Then, the lexicographic optimization on the team-slot side (\Cref{line:lexicographic-optimal} of the algorithm) can be improved by swapping $p$ and $p'$, a contradiction.

\vspace{2mm}

\emph{Participant-justified EF:} We first prove the following lemma.

\begin{lemma}
\label{lem:successively-better}
Over the course of \Cref{alg:main}, each team-slot only receives weakly increasing value.
\end{lemma}

\begin{innerproof}[Proof of Lemma~\ref{lem:successively-better}]
Assume for contradiction that in some iteration of the while-loop in \Cref{line:augment-eligibility}, at least one team-slot receives lower value than before.
Let $x$ be the first such team-slot, and $X$ be the set of team-slots before $x$.
Since $x$ receives lower value, the participant $p$ assigned to $x$ in the previous iteration must be assigned to a team-slot $x'\in X$ in the current iteration.

Due to the choice of $x$, the team-slot $x'$ cannot receive lower value than in the previous iteration.
Also, $p$ was assigned in the previous iteration, so its eligibility set did not expand between the two iterations.
By the lexicographic optimality of the matching in the previous iteration, $x'$ cannot receive higher value in the current iteration than in the previous one.
Hence, $x'$ receives the same value in the two iterations.

If the participant $p'$ assigned to $x'$ in the previous iteration is assigned to some team-slot $x''\in X$ in the current iteration, we can apply a similar argument to show that $x''$ must receive the same value in the two iterations.
In particular, if $x''$ receives higher value in the current iteration, then we could improve the previous iteration by moving $p'$ from $x'$ to $x''$ and $p$ from $x$ to $x'$, a contradiction with the lexicographic optimality.

We continue this chain of argument until eventually, we arrive at a team-slot $x^*$ such that the participant $p^*$ assigned to $x^*$ in the previous iteration is not assigned to a team-slot in $X$ in the current iteration.
Suppose that the chain of team-slots that we encounter is $x, x', x'', \dots, x^*$.
Modify the matching in the current iteration by reassigning backward along this chain---that is, move $p'$ from $x'$ to $x$, $p''$ from $x''$ to $x'$, and so on; moreover, assign $p^*$ to $x^*$.
Each team-slot in $X$ receives the same value as before, while $x$ receives higher value.
This contradicts the lexicographic optimality of the current iteration.
\end{innerproof}

We now return to proving participant-justified EF.
Assume for contradiction that participant~$p$ assigned to team-slot $x$ has justified envy toward participant~$p'$ assigned to team-slot $x'$.
In particular, $p$ strictly prefers $x'$ to $x$.
Hence, in some earlier iteration of the while-loop, the team corresponding to $x'$ was a least-preferred team in the eligibility set of $p$, and the eligibility set was expanded.
This means that $p$ was unmatched in that iteration, so $x'$ received a participant that it values at least as much as $p$.
By Lemma~\ref{lem:successively-better}, $x'$ must receive a participant that it values at least as much as $p$ in every later iteration.
However, $x'$ later receives $p'$, which it values less than~$p$---otherwise the envy of $p$ would not be justified.
This yields the desired contradiction.

\vspace{2mm}

\emph{PO:} To establish PO, we will use the concept of a Pareto improvement cycle due to \citet{ErdilEr17}.
Given an allocation, a cycle of participants $p_{j_1}, p_{j_2}, \dots, p_{j_k} \equiv p_{j_0}$ for some $k\ge 2$ is said to be a \emph{Pareto improvement cycle} if it has the property that for each $t\in \{0, 1, \dots, k-1\}$, participant $p_{j_t}$ weakly prefers participant $p_{j_{t+1}}$'s team to her own team and participant $p_{j_{t+1}}$'s team weakly prefers $p_{j_t}$ to $p_{j_{t+1}}$, and at least one of these preferences is strict for some $t$.
Theorem~1 of Erdil and Ergin implies the following characterization of participant-justified EF and PO allocations in our setting.

\begin{lemma}[\citep{ErdilEr17}]
\label{lem:erdil-ergin}
A participant-justified EF allocation satisfies PO if and only if it does not admit a Pareto improvement cycle.
\end{lemma}

Since the allocation returned by \Cref{alg:main} is participant-justified EF, by Lemma~\ref{lem:erdil-ergin}, it suffices to show that the final matching $S$ does not admit a Pareto improvement cycle.
This is indeed the case: A Pareto improvement cycle would imply that there exists a different allocation such that every party (i.e., team or participant) is weakly better off relative to~$S$ and at least one party is strictly better off.
This would mean that $S$ can be lexicographically improved on either the team-slot side or the participant side, contradicting the definition of~$S$.
Hence, the returned allocation satisfies PO.
\end{proof}

Besides fulfilling several desirable properties, our algorithm can also be implemented efficiently.

\begin{proposition}
\Cref{alg:main} can be implemented in polynomial time.    
\end{proposition}

\begin{proof}
Consider the while-loop in \Cref{line:augment-eligibility}, and note that at least one team is added to a participant's eligibility set in each iteration.
Hence, the while-loop runs for at most $O(nm)$ iterations.
In each iteration, to compute a lexicographically optimal matching $S$, we first determine the value that each team-slot receives according to the round-robin order.

Suppose that we are currently determining the value for team-slot $x$.
We construct an unweighted bipartite graph with team-slots on one side and participants on the other side.
For each team-slot before $x$, we add edges only to eligible participants that yield the determined value for the team-slot.
We then add edges from $x$ to eligible participants that yield the highest value for it, and compute a maximum matching in this graph.
If the maximum matching assigns a participant to all teams up to (and including) $x$, we fix the value of $x$ to be its value for the assigned participant, and proceed to the next team-slot.
Otherwise, we remove all current edges from $x$, and instead add edges from $x$ to eligible participants that yield the next highest value for it.
If this procedure exhausts all eligible participants for $x$ without assigning a participant to $x$, then the value of $x$ is determined to be $-\infty$.
That is, $x$ is unmatched in the current iteration, and we proceed to the team-slot after $x$.

With all values for team-slots determined, we only keep edges to eligible participants that yield the determined values.
We then proceed to optimize on the participant side in a similar manner.
Since computing a maximum matching can be done in polynomial time, and the numbers of teams and participants are polynomial, the algorithm can be implemented in polynomial time.
\end{proof}

We consider \Cref{alg:main} in two important special cases.
\begin{itemize}
\item Suppose that all participants are indifferent between all teams, as in the canonical (one-sided) fair division setting.
Then, every team is eligible for every participant from the beginning, and \Cref{alg:main} returns an allocation output by the round-robin algorithm.
\item Suppose that all teams and participants have strict preferences.
We will show that the output of \Cref{alg:main} coincides with that of the (participant-proposing) Gale--Shapley algorithm.
First, we claim that in each iteration of the while-loop, each participant can only be matched to her least-preferred team within her eligibility set.
Indeed, whenever a participant $p$ expands her eligibility set, she is unmatched, which means that the least-preferred team $i$ before the expansion is already matched to a participant that it prefers to $p$ (otherwise, by \Cref{line:lexicographic-optimal} of the algorithm, $i$ should be matched to $p$ instead).
By Lemma~\ref{lem:successively-better}, $i$ can never be matched to~$p$ later, so $p$ can only be matched to her least-preferred team after the expansion.
This implies that in each iteration of the while-loop, no two teams compete for the same participant and the round-robin ordering of team-slots is irrelevant.
In other words, each participant proposes to her least-preferred eligible team, each team chooses its most-preferred participants up to its number of slots, and each unmatched participant expands its eligibility set by adding her next-preferred team.
This is exactly what the Gale--Shapley algorithm does.
\end{itemize}

Hence, \Cref{alg:main} generalizes both round-robin and Gale--Shapley at once, thereby unifying fundamental algorithms from both fair division and two-sided matching.

We remark that in addition to team-justified SD-EF1, participant-justified EF, and balancedness, one cannot hope to also obtain PO for participants subject to the latter two properties.\footnote{This is similar to an optimality notion considered by \citet{ErdilEr17}.}
Indeed, consider an instance with $n = 2$ teams and $m = 4$ participants.
Each team has the preference $(p_1,p_2)\succ (p_3,p_4)$, participants $p_1,p_2$ are indifferent between both teams, while participants $p_3,p_4$ prefer team~$2$ to team~$1$.
The only allocation that is PO for participants subject to participant-justified EF and balancedness assigns $p_1,p_2$ to team~$1$ and $p_3,p_4$ to team~$2$.
However, this allocation is not team-justified SD-EF1.

\section{Group-Strategyproofness for Participants}
\label{sec:SP}

We next show the resemblance between \Cref{alg:main} and Gale--Shapley in that each participant has no incentive to misreport her preference, and there does not even exist a group of participants with an incentive to collude and report falsely.
That is, \Cref{alg:main} is group-strategyproof for participants.

For this result, we will need an additional tie-breaking criterion when choosing a lexicographically optimal matching $S$ in \Cref{line:lexicographic-optimal}.
Namely, denoting by $L$ the set of all team-slots, we create an arbitrary ordering of all $m^2$ elements in $P\times L$.
If there is more than one lexicographically optimal matching according to the criteria of \Cref{alg:main}, we choose one that is lexicographically optimal with respect to the order over $P\times L$---note that this yields a unique matching.
We refer to the resulting algorithm as \Cref{alg:main} with ``augmented tie-breaking''.

\begin{definition}
An algorithm that returns an allocation for any instance is 
{\em group-strategyproof (for participants)} if no group of participants can collude to misreport their preferences in a way that makes every of its members strictly better off.

Formally, let $I$ be any instance, $X$ be any non-empty subset of participants in $I$, and $I'$ be an instance obtained from $I$ by replacing the preference list $\succsim_p$ of each participant $p\in X$ with another (possibly identical) list $\succsim'_p$. Let $A$ and $A'$ be the allocations returned by the algorithm for inputs $I$ and $I'$, respectively. Then, for a group-strategyproof algorithm, it must hold that at least one participant $p\in X$ is assigned to a weakly better team in $A$ than in $A'$ with respect to her original preference $\succsim_p$.
\end{definition}

An algorithm is simply called {\em strategyproof (for participants)} if we consider only a singleton set as the set $X$ in the above definition. Clearly, group-strategyproofness is a stronger condition than strategyproofness.

\begin{theorem}\label{thm:SP}
\Cref{alg:main} with augmented tie-breaking is group-strategyproof for participants.
\end{theorem}
Before proving this theorem, we need to prepare several terms and observations.
Recall that in \Cref{alg:main}, an allocation is computed as a one-to-one matching between participants and team-slots, i.e., it is a subset of $P\times L$. 
We regard each participant $p$'s preference as a preference over team-slots; if $p$ strictly prefers team $i$ to team $j$ in the input preference, we consider that $p$ strictly prefers any of $i$'s slots to any of $j$'s slots. 
Slots corresponding to equally-preferred teams are equally preferred. 
Additionally, we assume that each slot has a preference list identical to that of the corresponding team. 
For a matching $S\subseteq P\times L$ and a participant $p$ (resp., a slot $x$) matched in $S$, we denote by $S(p)\in L$ (resp., $S(x)\in P$) the assigned partner in $S$.

Given a subset $E_j$ of $T$ for each participant $p_j\in P$, we say that a pair $(p_j,x)\in P\times L$ is {\em eligible} for $\{E_j\}_{j\in[m]}$ if the slot $x$ corresponds to a team in $E_j$. For a matching $S\subseteq P\times L$ consisting of eligible pairs, a sequence $R=(p_{j_1},x_1,p_{j_2},x_2, \dots,p_{j_h},x_h)$ with $h\geq 1$ is a {\em blocking path for $S$} ({\em within $\{E_j\}_{j\in[m]}$})  if the following properties are satisfied (see \Cref{fig:R}):
\begin{itemize}
\item All parties in the sequence are distinct. Moreover, $(p_{j_\ell},x_\ell)\not\in S$  for each $\ell\in\{1,2,\dots,h\}$ and  $(p_{j_\ell},x_{\ell-1})\in S$ for each $\ell\in\{2,3,\dots, h\}$. 
\item The first participant $p_{j_1}$ is either
\begin{enumerate}[(i)]
\setlength{\leftskip}{5mm}
\item unmatched in $S$ while $(p_{j_1}, x_1)$ is eligible for $\{E_j\}_{j\in[m]}$, or
\item matched in $S$ and strictly prefers $x_1$ to $S(p_{j_1})$.
\end{enumerate}
\item  For each $\ell\in\{1,2,\dots,h-1\}$, the slot $x_\ell$ prefers $p_{j_\ell}$ and $p_{j_{\ell+1}}$ equally. For each  $\ell\in\{2,\dots,h\}$, the participant $p_{j_{\ell}}$  prefers $x_{\ell-1}$ and $x_{\ell}$ equally.
\item The last slot $x_h$ is either
\begin{enumerate}[(i)]
\setlength{\leftskip}{5mm}
\item unmatched in $S$, 
\item matched in $S$ and strictly prefers   $p_{j_h}$ to $S(x_h)$, or
\item matched in $S$ and prefers $p_{j_h}$ and $S(x_h)$ equally, and moreover the participant $S(x_h)$ has an index larger than $j_1$.
\end{enumerate}
\end{itemize}
Note that for any eligibility sets $\{E_j\}_{j\in[m]}$ defined in \Cref{alg:main}, since $S$ consists only of eligible pairs and $p_{j_\ell}$ prefers $x_{\ell-1}$ and $x_\ell$ equally for $\ell \in \{2,\dots,h\}$, the edges $(p_{j_\ell}, x_\ell)$ are eligible for all $\ell \in \{2,\dots,h\}$.
Moreover, the second bullet point ensures that $(p_{j_1}, x_1)$ is eligible as well.
Observe also that in the case without ties, a blocking path reduces to an eligible pair that is a blocking pair in the classical sense.

\begin{figure}[htbp]
\begin{minipage}[t]{.58\textwidth}
\begin{minipage}[t]{.6\textwidth}
\centering
\begin{tikzpicture}[xscale=1.5,yscale=.8,very thick,font=\small]
\foreach \y [count=\i] in {9, 8, ..., 0}
{
    \coordinate(p\i) at (0,\y);
    \coordinate(s\i) at (1,\y);
}
\foreach \x / \y in {p2/s2, p3/s3, p4/s4, p5/s5, p6/s6}
{
    \draw[myred] (\x) -- (\y);
}

\draw[myred,densely dashed] (p7) -- (s7);

\foreach \x / \y in {p1/s2, p2/s3, p3/s4, p4/s5, p5/s6, p6/s7}
{
    \draw[mygray] (\x) -- (\y);
}
\foreach \a / \l in {p2/\sim, p3/\sim, p4/\sim, p5/\sim, p6/\sim}
{
    \node[font=\tiny,rotate=-120] at ($(\a)+(0.3,-.13)$) {$\boldsymbol{\l}$};
}
\foreach \a / \l in {s2/\sim, s3/\sim, s4/\sim, s5/\sim, s6/\sim, s7/\precsim}
{
    \node[font=\tiny,rotate=60] at ($(\a)+(-0.3,.13)$) {$\boldsymbol{\l}$};
}
\foreach \a in {p2, p3, p4, p5, s3, s4, s5, s6}
{
    \node[dot] at (\a) {};
}
\node[dot=6pt,fill=myred,draw=myred] at (p1) {};
\node[dot,label=left:{\scriptsize$p^*=p_{j_1}$}] at (p1) {};
\node[dot,label=left:{\scriptsize$\textcolor{myred}{S}(x_1)=p_{j_2}$}] at (p2) {};
\node[dot,label=left:{\scriptsize$\textcolor{myred}{S}(x_2)=p_{j_3}$}] at (p3) {};
\node[dot,label=left:{\scriptsize$\textcolor{myred}{S}(x_{h-1})=p_{j_h}$}] at (p6) {};
\node[dot,label=right:{\scriptsize$x_1$}] at (s2) {};
\node[dot,label=right:{\scriptsize$x_2$}] at (s3) {};
\node[dot,label=right:{\scriptsize$x_{h-1}$}] at (s6) {};
\node[dot,label=right:{\scriptsize$x_h$}] at (s7) {};

\begin{scope}[transform canvas={yshift=2pt},transparency group,opacity=.9]
    \draw[path] (p1) -- (s2) -- (p2)  -- (s3) -- (p3) -- (s4) -- (p4) -- (s5) -- (p5) -- (s6) -- (p6) -- (s7);
    \draw[overpath] (p1) -- (s2) -- (p2)  -- (s3) -- (p3) -- (s4) -- (p4) -- (s5) -- (p5) -- (s6) -- (p6) -- (s7);
\end{scope}
\node[mygreen,font=\small] at (0.5,9) {$R$};
\end{tikzpicture}
\subcaption{$p_{j_1}$ is unmatched in $S$}\label{subfig:Runmatch}
\end{minipage}%
\begin{minipage}[t]{.4\textwidth}
\centering
\begin{tikzpicture}[xscale=1.5,yscale=.8,very thick,font=\small]
\foreach \y [count=\i] in {9, 8, ..., 0}
{
    \coordinate(p\i) at (0,\y);
    \coordinate(s\i) at (1,\y);
}
\foreach \x / \y in {p1/s1, p2/s2, p3/s3, p4/s4, p5/s5, p6/s6}
{
    \draw[myred] (\x) -- (\y);
}

\draw[myred,densely dashed] (p7) -- (s7);

\foreach \x / \y in {p1/s2, p2/s3, p3/s4, p4/s5, p5/s6, p6/s7}
{
    \draw[mygray] (\x) -- (\y);
}
\foreach \a / \l in {p1/\prec, p2/\sim, p3/\sim, p4/\sim, p5/\sim, p6/\sim}
{
    \node[font=\tiny,rotate=-120] at ($(\a)+(0.3,-.13)$) {$\boldsymbol{\l}$};
}
\foreach \a / \l in {s2/\sim, s3/\sim, s4/\sim, s5/\sim, s6/\sim, s7/\precsim}
{
    \node[font=\tiny,rotate=60] at ($(\a)+(-0.3,.13)$) {$\boldsymbol{\l}$};
}
\foreach \a in {p1, p2, p3, p4, p5, p6, s2, s3, s4, s5, s6}
{
    \node[dot] at (\a) {};
}

\node[dot,label=left:{\scriptsize$p^*=p_{j_1}$}] at (p1) {};
\node[dot,label=left:{\scriptsize$p_{j_2}$}] at (p2) {};
\node[dot,label=left:{\scriptsize$p_{j_3}$}] at (p3) {};
\node[dot,label=left:{\scriptsize$p_{j_h}$}] at (p6) {};
\node[dot,label=right:{\scriptsize$\textcolor{myred}{S}(p_{j_1})$}] at (s1) {};
\node[dot,label=right:{\scriptsize$x_1$}] at (s2) {};
\node[dot,label=right:{\scriptsize$x_2$}] at (s3) {};
\node[dot,label=right:{\scriptsize$x_{h-1}$}] at (s6) {};
\node[dot,label=right:{\scriptsize$x_h$}] at (s7) {};

\begin{scope}[transform canvas={yshift=2pt},transparency group,opacity=.9]
    \draw[path] (p1) -- (s2) -- (p2)  -- (s3) -- (p3) -- (s4) -- (p4) -- (s5) -- (p5) -- (s6) -- (p6) -- (s7);
    \draw[overpath] (p1) -- (s2) -- (p2)  -- (s3) -- (p3) -- (s4) -- (p4) -- (s5) -- (p5) -- (s6) -- (p6) -- (s7);
\end{scope}
\node[mygreen,font=\small] at (0.8,8.7) {$R$};
\end{tikzpicture}
\subcaption{$p_{j_1}$ is matched in $S$}\label{subfig:Rmatch}
\end{minipage}%
\caption{Blocking path $R$}\label{fig:R}
\end{minipage}%
\hfill
\begin{minipage}[t]{.33\textwidth}
\centering
\begin{tikzpicture}[xscale=1.5,yscale=.8,very thick,font=\small]
\foreach \y [count=\i] in {9, 8, ..., 0}
{
    \coordinate(p\i) at (0,\y);
    \coordinate(s\i) at (1,\y);
}
\foreach \x / \y in {p1/s1, p2/s2, p3/s3, p4/s4, p5/s5, p6/s6}
{
    \draw[myred] ($(\x)+(0,.03)$) -- ($(\y)+(0,.03)$);
}
\draw[myred,densely dashed] ($(p7)+(0,.04)$) -- ($(s7)+(0,.04)$);

\foreach \x / \y in {p2/s2, p3/s3, p4/s4, p5/s5, p6/s6}
{
    \draw[myblue] ($(\x)+(0,-.04)$) -- ($(\y)+(0,-.04)$);
}
\draw[myblue,densely dashed] ($(p7)+(0,-.04)$) -- ($(s7)+(0,-.04)$);

\foreach \x / \y in {p1/s2, p2/s3, p3/s4, p4/s5, p5/s6, p6/s7}
{
    \draw[mygray] (\x) -- (\y);
}

\foreach \a / \l in {p1/\prec, p2/\sim, p3/\sim, p4/\sim, p5/\sim, p6/\sim}
{
    \node[font=\tiny,rotate=-120] at ($(\a)+(0.3,-.13)$) {$\boldsymbol{\l}$};
}
\foreach \a / \l in {s2/\sim, s3/\sim, s4/\sim, s5/\sim, s6/\sim, s7/\precsim}
{
    \node[font=\tiny,rotate=60] at ($(\a)+(-0.3,.13)$) {$\boldsymbol{\l}$};
}

\foreach \a in {p2, p3, p4, p5, s3, s4, s5, s6}
{
    \node[dot] at (\a) {};
}
\node[dot,label=left:{\scriptsize$p^*=p_{j_1}$}] at (p1) {};
\node[dot,label=left:{\scriptsize$p_{j_2}$}] at (p2) {};
\node[dot,label=left:{\scriptsize$p_{j_3}$}] at (p3) {};
\node[dot,label=left:{\scriptsize$p_{j_h}$}] at (p6) {};
\node[dot,label=right:{\scriptsize$\textcolor{myred}{S}(p_{j_1})$}] at (s1) {};
\node[dot,label=right:{\scriptsize$x_1$}] at (s2) {};
\node[dot,label=right:{\scriptsize$x_2$}] at (s3) {};
\node[dot,label=right:{\scriptsize$x_{h-1}$}] at (s6) {};
\node[dot,label=right:{\scriptsize$x_h$}] at (s7) {};

\begin{scope}[transform canvas={yshift=2pt},transparency group,opacity=.9]
    \draw[path] (p1) -- (s2) -- (p2)  -- (s3) -- (p3) -- (s4) -- (p4) -- (s5) -- (p5) -- (s6) -- (p6) -- (s7);
    \draw[overpath] (p1) -- (s2) -- (p2)  -- (s3) -- (p3) -- (s4) -- (p4) -- (s5) -- (p5) -- (s6) -- (p6) -- (s7);
\end{scope}
\node[mygreen,font=\small] at (0.8,8.7) {$R$};
\end{tikzpicture}
\caption{$S'(x_i)=S(x_i)~(i=1,\dots,h)$}\label{fig:S'=S}
\end{minipage}
\end{figure}

For the proof of group-strategyproofness, we show the following properties of \Cref{alg:main}.
\begin{proposition}
\label{prop:noblocking}
Consider \Cref{alg:main} with augmented tie-breaking.
For each iteration of the while-loop in \Cref{line:augment-eligibility}, 
\begin{itemize}
\item[--] the computed matching $S$ admits no blocking path within the eligibility sets $\{E_j\}_{j\in [m]}$, and
\item[--] any participant that is matched in $S$ is matched to a slot corresponding to a least-preferred team in her eligibility set.
\end{itemize}
\end{proposition}
\begin{proof}
The statement of the proposition is equivalent to the statement that the following two claims hold for every participant $p_j$ with $j\in[m]$.  
\begin{enumerate}
\item[(a)] For each iteration of the while-loop in \Cref{line:augment-eligibility}, the computed matching $S$ admits no blocking path  within $\{E_{j'}\}_{j'\in [m]}$ that starts at $p_j$.
\item[(b)] For each iteration of the while-loop in \Cref{line:augment-eligibility}, if $p_j$ is matched in the computed matching $S$, then she is matched to a slot corresponding to a least-preferred team in $E_j$.
\end{enumerate}
Suppose for contradiction that there exist participants violating (a) or (b). Among such participants, let $p^*$ be the one with the smallest index. 
We consider two cases depending on whether $p^*$ violates (a) or (b) (if $p^*$ violates both, we choose one case arbitrarily).

\medskip
\noindent\textbf{Case 1. } We first assume that $p^*$ violates (a).
Then, at some iteration, the computed matching~$S$ has a blocking path $R=(p_{j_1},x_1,p_{j_2},x_2, \dots,p_{j_h},x_h)$ within $\{E_j\}_{j\in[m]}$ where $p_{j_1}=p^*$. 
By the definition of a blocking path, $p_{j_1}$ is either unmatched in $S$ or strictly prefers $x_1$ to $S(p_{j_1})$.

First, assume that $p_{j_1}$ is unmatched in $S$ (see \Cref{subfig:Runmatch}).  
Consider modifying $S$ by flipping along $R$ (i.e., adding $\{(p_{j_\ell}, x_\ell)\}_{\ell\in\{1,\dots,h\}}$ to $S$ and removing $\{(p_{j_\ell}, x_{\ell-1})\}_{\ell\in\{2,\dots,h\}}$ from $S$) and, in addition, removing the pair in $S$ covering $x_h$ if such a pair exists. 
The resultant matching $\tilde{S}$ consists of edges eligible for $\{E_{j}\}_{j\in [m]}$. 
Furthermore, compared to $S$, in this matching $\tilde{S}$, the first participant $p_{j_1}=p^*$ is better off (turns from unmatched to matched), all intermediate parties on~$R$ are assigned equally good partners, and either (i) the last slot $x_h$ is strictly better off or (ii) $x_h$ is assigned an equally good participant and the participant $S(x_h)$, whose index is larger than that of $p_{j_1}=p^*$, is worse off (turns from matched to unmatched). 
In both cases, the matching $\tilde{S}$ is better than $S$ with respect to the criteria used in the algorithm, which contradicts the fact that $S$ is chosen.

Therefore, we assume that $p_{j_1}=p^*$ is matched in $S$ and strictly prefers $x_1$ to $S(p_{j_1})$ (see \Cref{subfig:Rmatch}). 
This means that, in some iteration earlier than the computation of $S$, the participant~$p_{j_1}$ was left unmatched while $(p_{j_1}, x_1)$ was eligible at that time. 
Let $S'$ be the matching computed in that iteration. 
Among the slots on the path $R= (p_{j_1}, x_1, p_{j_2}, x_2, \dots, p_{j_h}, x_h)$, let $x^*$ be the first one for which $S'(x^*) \neq S(x^*)$ (where we regard $S'(x^*) = S(x^*)$ if and only if $x^*$ is matched to the same participant in $S'$ and $S$ or is unmatched in both). 
We claim that such a slot must exist; suppose for contradiction that it does not (see \Cref{fig:S'=S}).
Then, $S'(x_\ell) = S(x_\ell)$ for every slot $x_\ell$ on $R$, which means that $R$ forms a blocking path for $S'$ with $p_{j_1}$ being unmatched in $S'$. 
It follows that we can modify $S'$ in the same manner as $S$ was modified to $\tilde{S}$ above, resulting in a matching that is better than $S'$ and consists only of pairs that were eligible when $S'$ was computed. 
This contradicts the fact that $S'$ was chosen in that iteration.
Hence, the required $x^*$ must exist. 
Let $R[p_{j_1}, x^*]$ be the subpath of $R$ from $p_{j_1}$ to $x^*$. 
Observe that $x^*$ is matched in $S'$; otherwise, we could improve $S'$ by flipping along $R[p_{j_1}, x^*]$.
Moreover, $x^*$ is also matched in $S$ by Lemma~\ref{lem:successively-better}.

Consider an undirected bipartite graph $G$ whose vertex set is $P\cup L$ and edge set is $(S\setminus S')\cup (S'\setminus S)$ (where the edges are regarded as undirected edges). Take a connected component in $G$ containing~$x^*$. 
Since $S'(x^*)\neq S(x^*)$ and $x^*$ is matched in $S'$, the component is a path or a cycle, in which $x^*$ is covered by an $S'$-edge (see \Cref{fig:noblocking-alt}). 
Let $Q$ be a maximal path/cycle in the component starting at $x^*$ with the edge $(x^*,S'(x^*))$. 
Let $q$ be the participant on $Q$ closest to $x^*$ subject to the condition that $q$ is either unmatched in $S$ or satisfies $S'(q)\succ_q S(q)$ (see \Cref{subfig:noblocking-path}; if $Q$ is a cycle, we view it as the directed cycle when determining the closest $q$).
If there is no such participant on $Q$, let $q$ be the last participant on $Q$ (in case $Q$ is a cycle, this means that $q=S(x^*)$), 
which implies that $Q$ ends at a slot (\Cref{subfig:noblocking-end}) or $Q$ is a cycle (\Cref{subfig:noblocking-cycle}) in this case.
Let $W=R[p_{j_1}, x^*]+Q[x^*,q]$. 
We show the following properties of $W$.

\begin{figure}[htbp]
\input{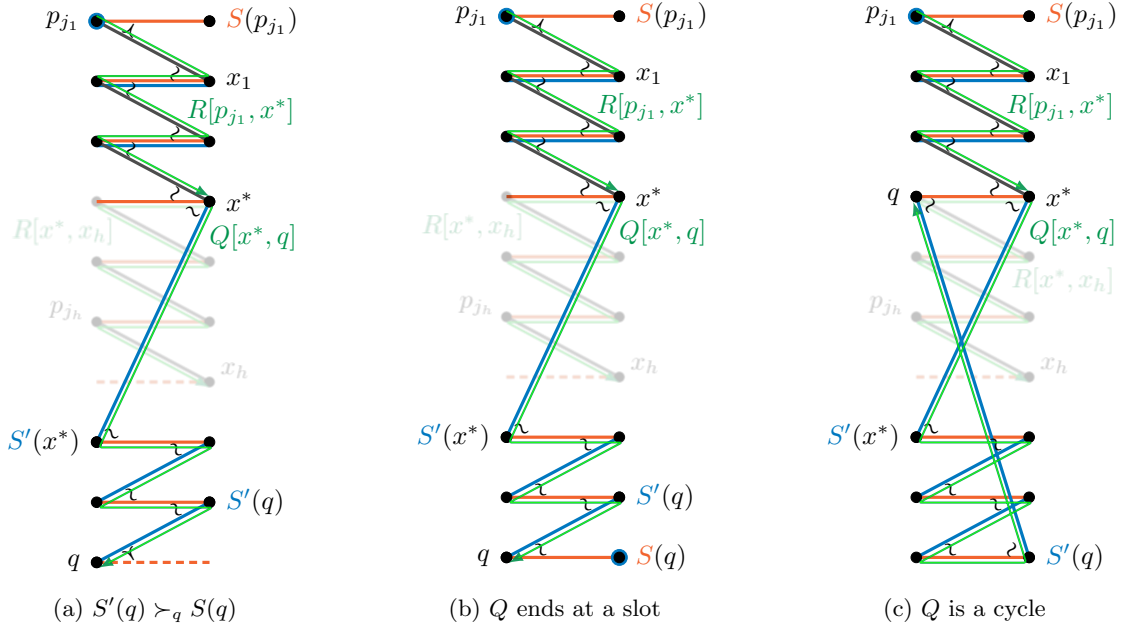}
\caption{Path $W=R[p_{j_1},x^*]+Q[x^*,q]$. Edges in $S$ and $S'$ are depicted in red and blue, respectively. Note that $Q[x^*,q]$ may intersect with $R[x^*,x_h]$ (which is faded in these figures) at points other than $x^*$.}\label{fig:noblocking-alt}
\end{figure}

\begin{enumerate}[(1)]
\item $W$ is a path that starts at a participant unmatched in $S'$ and uses pairs not in $S'$ and those in~$S'$ alternately.
\item Every participant on $W$ other than $p_{j_1}=p^*$ has an index smaller than $j_1$. In particular, $q$ does.
\item Every intermediate party on $W$ prefers its preceding and succeeding parties equally.
In addition, $x^*$ prefers the preceding participant, $S'(x^*)$, and $S(x^*)$ all equally.
\item The last participant $q$ is either unmatched in $S$ or strictly prefers $S'(q)$ to $S(q)$.
\end{enumerate}
We show these properties in turn.

\smallskip
\noindent\textit{Property (1). }
By the definition of $S'$, the participant $p_{j_1}=p^*$ is unmatched in $S'$. By the definition of $x^*$, all the intermediate parties on $R[p_{j_1}, x^*]$ are covered by edges in $S'\cap S$. Since $Q[x^*,q]$ is a path that starts with an $S'$-edge and uses $S'$-edges and $S$-edges alternately, all parties in $Q[x^*,q]$ are covered by edges in $S'\setminus S$. Thus, all the parties on $W$ are distinct, and $W$ satisfies property~(1).

\smallskip
\noindent\textit{Properties (2)--(3). }
As for the intermediate parties on $R[p_{j_1},x^*]$, property (3) immediately follows from the definition of a blocking path. In addition, for any participant $p\neq p_{j_1}$ on $R[p_{j_1}, x^*]$, property (2) holds; otherwise, we could improve $S'$ by flipping along $R[p_{j_1}, p]$, which would make $p_{j_1}$ matched and $p$ unmatched.

We next show property (3) for the slot $x^*$. 
Since $x^*$ is on a blocking path $R$, it weakly prefers the participant preceding it in $W$, say $\hat{p}$, to the participant $S(x^*)$ or is unmatched in $S$. Furthermore, as we have shown, $x^*$ is matched in $S'$. 
As $S'$ is computed earlier than $S$, Lemma~\ref{lem:successively-better} implies that $x^*$ is matched in $S$ and $S(x^*)\succsim_{x^*}S'(x^*)$. 
Thus, we have $\hat{p}\succsim_{x^*}S(x^*)\succsim_{x^*}S'(x^*)$. 
At the same time, we must have $S'(x^*) \succsim_{x^*} \hat{p}$; otherwise, we could improve $S'$ by flipping along $R[p_{j_1}, x^*]$ and removing $(S'(x^*), x^*)$, which would make $x^*$ strictly better off. 
It follows that $\hat{p}\sim_{x^*}S(x^*)\sim_{x^*}S'(x^*)$. Note also that $S'(x^*)$ is the party succeeding $x^*$ in $W$. 
Hence, property (3) holds for $x^*$. 

For the remaining parties, i.e., those on $Q[x^*,q]$ except $x^*$, we show properties (2) and (3) by induction.
Take a participant $p$ in $Q[x^*,q]$ and suppose that property (3) holds for all parties in $W$ preceding $p$. 
The slot preceding $p$ in $W$ is $S'(p)$ and the one succeeding $p$ is $S(p)$ if it exists. 
The participant $p$ must have an index smaller than $j_1$; otherwise, we could improve $S'$ by flipping along $W[p_{j_1},p]=R[p_{j_1}, x^*]+Q[x^*,p]$. 
Since $p$ has an index smaller than $j_1$, by the choice of $p^*=p_{j_1}$, we have that $p$ satisfies condition (b), i.e., $p$ is matched to a slot of a least-preferred eligible team in~$S$. 
Since the eligibility set is monotone increasing and $S'$ is computed earlier than $S$, this implies that $S'(p)\succsim_p S(p)$ whenever $p$ is matched in $S$. 
If $p\neq q$, then by the choice of $q$, the participant $p$ is matched in $S$ and $S'(p)\not\succ_p S(p)$, and hence we have $S'(p)\sim_p S(p)$. 
Thus, properties (2) and~(3) hold for $p$. 
In addition, when $p\neq q$, the slot $y\coloneqq S(p)$, which succeeds $p$ in $W$, is matched in $S'$ and satisfies $S'(y)\succsim_y S(y)=p$; otherwise, we could improve $S'$ by flipping along $W[p_{j_1}, S'(y)]$, which would make $y$ strictly better off without making other slots strictly worse off. 
At the same time, Lemma~\ref{lem:successively-better} implies that $S(y)\succsim_y S'(y)$. 
Thus, we have  $S'(y)\sim_{y}S(y)$, i.e., property (3) holds for the slot succeeding $p$ in $W$.

\smallskip
\noindent\textit{Property (4). }
Finally, we show that $q$ is either unmatched in $S$ or satisfies $S'(q)\succ_{q} S(q)$.
Suppose for contradiction that this does not hold. 
This means that $q$ is matched in $S$ and $S(q)\succsim_q S'(q)$.
By the argument in the previous paragraph, we also have $S'(q)\succsim_q S(q)$, so it must be that $S'(q)\sim_q S(q)$.
Moreover, by the definition of $q$, either $Q$ ends at the slot $S(q)$ that is unmatched in $S'$ or $Q$ is a cycle with $S(q)=x^*$. 

In the former case (\Cref{subfig:noblocking-end}), $W+(q,S(q))$ is an $S'$-alternating path starting and ending at parties unmatched in $S'$. Hence, we can improve $S'$ by flipping along  $W+(q,S(q))$, a contradiction.
In the latter case (\Cref{subfig:noblocking-cycle}), $Q$ is a cycle using edges in $S'$ and $S$ alternately and, by property (3), every party on $Q$  prefers her partners in $S$ and $S'$ equally. 
Then, all edges on $Q$ are eligible both when $S$ is computed and when $S'$ is computed. 
This contradicts the fact that matchings $S$ and $S'$ are chosen in the algorithm using the augmented tie-breaking.
\smallskip

\begin{figure}[t]
\begin{minipage}{.33\textwidth}
\centering
\begin{tikzpicture}[xscale=1.5,yscale=.8,very thick,font=\small]
\draw[opacity=0] (-.5,0) grid[step=0.5] (1.5,7);
\foreach \y [count=\i] in {9, 8, ..., 0}
{
    \coordinate(p\i) at (0,\y);
    \coordinate(s\i) at (1,\y);
}
\foreach \x / \y in {p4/s4, p5/s5, p6/s6, p8/s8, p9/s9}
{
    \draw[myred] (\x) -- (\y);
}
\draw[myred,densely dashed] (p7) -- (s7);
\draw[myred,densely dashed] (p10) -- (s10);
\foreach \x / \y in {s4/p8, s8/p9, s9/p10}
{
    \draw[myblue] (\x) -- (\y);
}

\foreach \x / \y in {p4/s5, p5/s6, p6/s7}
{
    \draw[mygray] (\x) -- (\y);
}

\foreach \a / \l in {p4/\sim, p5/\sim, p6/\sim}
{
    \node[font=\tiny,rotate=-120] at ($(\a)+(.3,-.13)$) {$\boldsymbol{\l}$};
}
\foreach \a / \l in {s5/\sim, s6/\sim}
{
    \node[font=\tiny,rotate=60] at ($(\a)+(-.3,.13)$) {$\boldsymbol{\l}$};
}
\node[font=\tiny,rotate=-40] at ($(s4)+(-.15,-.18)$) {$\boldsymbol{\sim}$};
\node[font=\tiny,rotate=-60] at ($(s8)+(-.3,.-.13)$) {$\boldsymbol{\sim}$};
\node[font=\tiny,rotate=-60] at ($(s9)+(-.3,.-.13)$) {$\boldsymbol{\sim}$};
\node[font=\tiny,rotate=-40] at ($(p8)+(.13,.13)$) {$\boldsymbol{\sim}$};
\node[font=\tiny,rotate=-60] at ($(p9)+(.3,.13)$) {$\boldsymbol{\sim}$};
\node[font=\tiny,rotate=-60] at ($(p10)+(.3,.13)$) {$\boldsymbol{\succ}$};

\foreach \a in {p4, p5, p6, p8, p9, p10, s5, s6, s7, s8, s9}
{
    \node[dot] at (\a) {};
}
\node[label={[left]:$p_{j_h}$}] at (p6) {};
\node[dot,label=left:{$\textcolor{myblue}{S'}(\hat{x})$}] at (p8) {};
\node[dot,label=left:{$q$}] at (p10) {};
\node[dot,label=right:{$\hat{x}$}] at (s4) {};
\node[label={[right]:$x_h$}] at (s7) {};
\node[dot,label=right:{$\textcolor{myblue}{S'}(q)$}] at (s9) {};

\begin{scope}[transform canvas={yshift=2pt},transparency group,opacity=.9]
    \draw[path] (s4) -- (p4) -- (s5) -- (p5) -- (s6) -- (p6) -- (s7);
    \draw[overpath] (s4) -- (p4) -- (s5) -- (p5) -- (s6) -- (p6) -- (s7);
\end{scope}
\begin{scope}[transform canvas={yshift=-1.5pt,xshift=1pt},transparency group,opacity=.9]
    \draw[path] (p10) -- (s9) -- (p9) -- (s8) -- (p8) -- (s4);
    \draw[overpath] (p10) -- (s9) -- (p9) -- (s8) -- (p8) -- (s4);
\end{scope}
\node[mygreen,font=\small,left] at (0.3,6.5) {$R[\hat{x},x_h]$};
\node[mygreen,font=\small,right] at (0.95,5.5) {$\overline{Q}[q,\hat{x}]$};
\end{tikzpicture}
\subcaption{$\hat{x}\ne x_h$ (possibly $\hat{x}=x^*)$}\label{subfig:pathWnormal}
\end{minipage}%
\begin{minipage}{.33\textwidth}
\centering
\begin{tikzpicture}[xscale=1.5,yscale=.8,very thick,font=\small]
\draw[opacity=.0] (-.5,0) grid[step=0.5] (1.5,7);
\foreach \y [count=\i] in {9, 8, ..., 0}
{
    \coordinate(p\i) at (0,\y);
    \coordinate(s\i) at (1,\y);
}
\foreach \x / \y in {p7/s7, p8/s8, p9/s9}
{
    \draw[myred] (\x) -- (\y);
}
\draw[myred,densely dashed] (p10) -- (s10);
\foreach \x / \y in {s7/p8, s8/p9, s9/p10}
{
    \draw[myblue] (\x) -- (\y);
}
\draw[mygray] (p6) -- (s7);

\node[font=\tiny,rotate=60] at ($(s7)+(-.3,.13)$) {$\boldsymbol{\sim}$};
\node[font=\tiny,rotate=-40] at ($(s7)+(-.3,-.13)$) {$\boldsymbol{\sim}$};
\node[font=\tiny,rotate=-60] at ($(s8)+(-.3,-.13)$) {$\boldsymbol{\sim}$};
\node[font=\tiny,rotate=-60] at ($(s9)+(-.3,-.13)$) {$\boldsymbol{\sim}$};
\node[font=\tiny,rotate=-60] at ($(p8)+(.3,.13)$) {$\boldsymbol{\sim}$};
\node[font=\tiny,rotate=-60] at ($(p9)+(.3,.13)$) {$\boldsymbol{\sim}$};
\node[font=\tiny,rotate=-60] at ($(p10)+(.3,.13)$) {$\boldsymbol{\succ}$};

\foreach \a in {p6, p7, p8, p9, p10, s8, s9}
{
    \node[dot] at (\a) {};
}
\node[label={[left]:$p_{j_h}$}] at (p6) {};
\node[dot,label=left:{$\textcolor{myred}{S}(x_h)$}] at (p7) {};
\node[dot,label=left:{$\textcolor{myblue}{S'}(\hat{x})$}] at (p8) {};
\node[dot,label=left:{$q$}] at (p10) {};
\node[dot,label=right:{$x_h=x^*$}] at (s7) {};
\node[dot,label=right:{$\textcolor{myblue}{S'}(q)$}] at (s9) {};

\begin{scope}[transform canvas={yshift=-2pt,xshift=1pt},transparency group,opacity=.9]
    \draw[path] (p10) -- (s9) -- (p9) -- (s8) -- (p8) -- (s7);
    \draw[overpath] (p10) -- (s9) -- (p9) -- (s8) -- (p8) -- (s7);
\end{scope}
\node[mygreen,font=\small,right] at (0.9,2.4) {$\overline{Q}[q,\hat{x}]$};
\end{tikzpicture}
\subcaption{$x^*=\hat{x}=x_h$}\label{subfig:pathWshort}
\end{minipage}
\begin{minipage}{.33\textwidth}
\centering
\begin{tikzpicture}[xscale=1.5,yscale=.8,very thick,font=\small]
\draw[opacity=.0] (-.5,0) grid[step=0.5] (1.5,7);
\foreach \y [count=\i] in {9, 8, ..., 0}
{
    \coordinate(p\i) at (0,\y);
    \coordinate(s\i) at (1,\y);
}
\coordinate(p4) at (0,6.5);
\coordinate(s4) at (1,6.5);

\foreach \x / \y in {p4/s4, p5/s5, p6/s6, p7/s7, p8/s8, p9/s9}
{
    \draw[myred] (\x) -- (\y);
}
\draw[myred,densely dashed] (p10) -- (s10);
\foreach \x / \y in {p7/s4, s7/p8, s8/p9, s9/p10}
{
    \draw[myblue] (\x) -- (\y);
}
\draw[myblue] ($(p4)+(0,-.02)$) -- ($(s5)+(0,-.02)$);
\draw[mygray] (p5) -- (s6);
\draw[mygray] (p6) -- (s7);

\node[font=\tiny,rotate=-60] at ($(s4)+(-.13,-.13)$) {$\boldsymbol{\sim}$};
\node[font=\tiny,rotate=60] at ($(s5)+(-.2,.13)$) {$\boldsymbol{\sim}$};
\node[font=\tiny,rotate=40] at ($(s6)+(-.3,.13)$) {$\boldsymbol{\sim}$};
\node[font=\tiny,rotate=60] at ($(s7)+(-.3,.13)$) {$\boldsymbol{\precsim}$};
\node[font=\tiny,rotate=-40] at ($(s7)+(-.3,-.13)$) {$\boldsymbol{\sim}$};
\node[font=\tiny,rotate=-60] at ($(s8)+(-.3,-.13)$) {$\boldsymbol{\sim}$};
\node[font=\tiny,rotate=-60] at ($(s9)+(-.3,-.13)$) {$\boldsymbol{\sim}$};

\node[font=\tiny,rotate=60] at ($(p4)+(.22,-.15)$) {$\boldsymbol{\sim}$};
\node[font=\tiny,rotate=60] at ($(p5)+(.22,-.13)$) {$\boldsymbol{\sim}$};
\node[font=\tiny,rotate=60] at ($(p6)+(.22,-.13)$) {$\boldsymbol{\sim}$};
\node[font=\tiny,rotate=-50] at ($(p7)+(.13,.13)$) {$\boldsymbol{\sim}$};
\node[font=\tiny,rotate=-60] at ($(p8)+(.3,.13)$) {$\boldsymbol{\sim}$};
\node[font=\tiny,rotate=-60] at ($(p9)+(.3,.13)$) {$\boldsymbol{\sim}$};
\node[font=\tiny,rotate=-60] at ($(p10)+(.3,.13)$) {$\boldsymbol{\succ}$};

\foreach \a in {p4, p5, p6, p7, p8, p9, p10, s4, s5, s6, s7, s8, s9}
{
    \node[dot] at (\a) {};
}
\node[dot,label=right:{$\tilde{x}$}] at (s5) {};

\node[label={[left]:$p_{j_h}$}] at (p6) {};
\node[dot,label=left:{$\textcolor{myred}{S}(x_h)$}] at (p7) {};
\node[dot,label=left:{$\textcolor{myblue}{S'}(\hat{x})$}] at (p8) {};
\node[dot,label=left:{$q$}] at (p10) {};
\node[dot,label=right:{$x_h=\hat{x}$}] at (s7) {};
\node[dot,label=right:{$\textcolor{myblue}{S'}(q)$}] at (s9) {};

\begin{scope}[transform canvas={yshift=-2pt,xshift=1pt},transparency group,opacity=.9]
    \draw[path,shorten <= 3pt] (s7) --(p7) -- (s4) -- ($(p4)+(.1,0)$) -- ($(s5)+(0,.15)$);
    \draw[overpath,shorten <= 3pt] (s7) --(p7) -- (s4) -- ($(p4)+(.1,0)$) -- ($(s5)+(0,.15)$);
\end{scope}
\node[mygreen,font=\small,right] at (0.9,5.8) {$\overline{Q}[\hat{x},\tilde{x}]$};
\begin{scope}[transform canvas={yshift=-1.5pt,xshift=0pt},transparency group,opacity=.9]
    \draw[path] (s5) -- (p5) -- (s6) -- ($(p6)+(.15,0)$) -- ($(s7)+(0,.15)$);
    \draw[overpath] (s5) -- (p5) -- (s6) -- ($(p6)+(.15,0)$) -- ($(s7)+(0,.15)$);
\end{scope}
\node[mygreen,font=\small,right] at (0.9,3.4) {$R[\tilde{x},\hat{x}]$};
\begin{scope}[transform canvas={yshift=-2pt,xshift=1pt},transparency group,opacity=.9]
    \draw[path,shorten >=3pt] (p10) -- (s9) -- (p9) -- (s8) -- (p8) -- (s7);
    \draw[overpath,shorten >=7pt] (p10) -- (s9) -- (p9) -- (s8) -- (p8) -- (s7);
\end{scope}
\node[mygreen,font=\small,right] at (0.9,2.4) {$\overline{Q}[q,\hat{x}]$};
\end{tikzpicture}
\subcaption{$x^*\ne \hat{x}=x_h$ (possibly $\tilde{x}=x^*$)}\label{subfig:pathWcycle}
\end{minipage}  
\caption{Path $\tilde{R}=\overline{Q}[q,\hat{x}]+R[\hat{x},x_h]$}\label{fig:Case1final}
\end{figure}
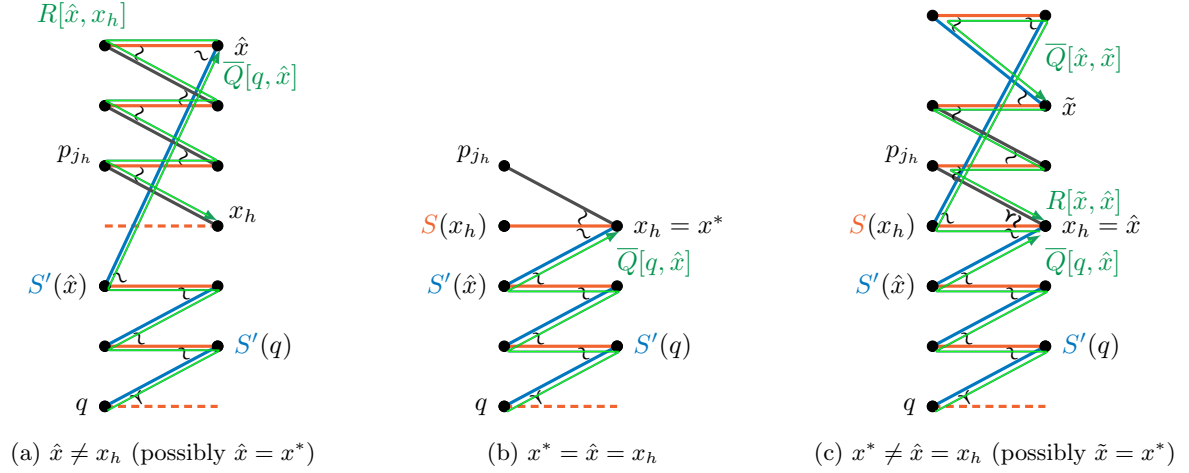

Thus, we have established properties (1)--(4) and have shown that, among the three figures in \Cref{fig:noblocking-alt}, only \Cref{subfig:noblocking-path} represents a possible situation.
Now, consider a sequence $\overline{Q}[q,x^*]+R[x^*, x_h]$, where $\overline{Q}[q,x^*]$ refers to the reverse subpath of $Q$ from $q$ to  $x^*$, 
and $R[x^*, x_h]$ is a subpath of $R$ from $x^*$ to $x_h$ (it is possible that $x^*=x_h$). 
Note that this sequence is not necessarily a path because $\overline{Q}[q,x^*]$ and $R[x^*, x_h]$ may share vertices other than $x^*$. Remove the maximal cyclic part if it exists and let $\tilde{R}$ be the resultant path.
More precisely, we let $\hat{x}$ be the first slot on the path $\overline{Q}[q,x^*]$ for which $\hat{x}\in R[x^*,x_h]$ (possibly $\hat{x}=x^*$) and then let $\tilde{R}=\overline{Q}[q,\hat{x}]+R[\hat{x},x_h]$ (see \Cref{fig:Case1final}).
We claim that $\tilde{R}$ is a blocking path for $S$ within the eligibility sets $\{E_j\}_{j\in [m]}$ at the time when $S$ is computed. 
\begin{itemize}
\item By its definition, $\tilde{R}$ uses edges not in $S$ and those in $S$ alternately. 
\item Since the first pair $(q, S'(q))$ of $\tilde{R}$ is eligible when $S'$ is computed,  it is also eligible when $S$ is computed. 
By property (4),  $q$ is either unmatched in $S$ or satisfies $S'(q)\succ_q S(q)$.
\item  As $R$ is a blocking path for $S$ and we have property (3), every intermediate party in $\tilde{R}$ prefers its preceding and succeeding parties equally. (One of them is the partner in $S$, and the other is equally good. This holds true even if the cyclic part is removed when defining $\tilde{R}$.)
\item As $R$ is a blocking path for $S$, either (i) $x_h$ is unmatched in $S$, (ii) $x_h$ is matched in $S$ and  $p_{j_h}\succ_{x_h}S(x_h)$, or (iii) $x_h$ is matched in $S$,  $p_{j_h}\sim_{x_h}S(x_h)$,  and $S(x_h)$ has an index larger than $j_1$. Note that, in $\tilde{R}$, the participant preceding $x_h$ is $p_{j_h}$ if $\hat{x}\neq x_h$ and $S'(x_h)$ if $\hat{x}=x_h$. 
\begin{itemize}
\item In the case where $\hat{x}\neq x_h$ (see \Cref{subfig:pathWnormal}), we can immediately see that $\tilde{R}$ is a blocking path if (i) or (ii) above holds. When (iii) holds for $R$, the index of $S(x_h)$ is larger than $j_1$. As property (2) implies that $q$ has an index smaller than $j_1$, we obtain that the index of $S(x_h)$ is larger than that of $q$. Thus, $\tilde{R}$ is a blocking path in this case. 

\item In the case where $x^*=\hat{x}=x_h$ (see \Cref{subfig:pathWshort}), the second claim in property (3) implies $p_{j_h}\sim_{x_h}S'(x_h)\sim_{x_h}S(x_h)$, where $S'(x_h)$ is the participant preceding $x_h$ in $\tilde{R}$. Then, (iii) above applies to $R$. By the same argument as above, property (2) implies that the index of $S(x_h)$ is larger than that of $q$. Hence, $\tilde{R}$ is a blocking path.

\item In the case where $x^*\neq\hat{x}=x_h$ (see \Cref{subfig:pathWcycle}), $\hat{x} = x_h$ belongs to $Q[x^*,q]$ and therefore to $W$.
Hence, property (3) implies $S(x_h)\sim_{x_h} S'(x_h)$, where $S'(x_h)$ is the participant preceding $x_h$ in $\tilde{R}$.
Let $\tilde{x}$ be the second slot on the path $\overline{Q}[q,x^*]$ for which $\tilde{x}\in R[x^*,x_h]$ (possibly $\tilde{x} = x^*$). Then, we cannot have $p_{j_h}\succ_{x_h}S(x_h)$; otherwise, we could improve~$S$ by flipping along the cycle $\overline{Q}[\hat{x},\tilde{x}]+R[\tilde{x},\hat{x}]$, a contradiction. Thus, we have $p_{j_h}\sim_{x_h}S(x_h)$. Then, again (iii) applies to $R$ and property (2) implies that the index of $S(x_h)$ is larger than $q$. This confirms that $\tilde{R}$ is a blocking path.
\end{itemize}
\end{itemize}
Therefore, $\tilde{R}$ is indeed a blocking path for $S$ starting at $q$. Since $q$ has an index smaller than $j_1$, this contradicts the minimality of the index of $p^*=p_{j_1}$.

\medskip
\noindent\textbf{Case 2. } We next assume that $p^*$ violates (b). 
Then, at some iteration, $p^*$ is matched in the computed matching $S$ and there is an eligible team strictly worse than the team corresponding to the slot $S(p^*)$. 
This means that, in some iteration earlier than the computation of $S$, the participant $p^*$ was left unmatched while $(p^*, S(p^*))$ was eligible. Let $S'$ be the matching computed in that iteration. 
Consider an undirected bipartite graph $G$ whose vertex set is $P\cup L$ and edge set is $(S\setminus S')\cup (S'\setminus S)$. 
Since $p^*$ is matched in $S$ but not matched in $S'$, the connected component containing $p^*$ is a path starting at $p^*$ with an $S$-edge. 
Denote this path by $R$.
Let $q$ be the participant on $R$ closest to $p^*$ subject to the condition that $q$ is either unmatched in $S$ or satisfies $S'(q)\succ_q S(q)$ (\Cref{subfig:case2normal}).
If there is no such participant on $R$ (which implies that $R$ ends at a slot), let $q$ be the last participant on $R$ (\Cref{subfig:case2end}).
Let $W$ be the subpath of $R$ from $p^*$ to $q$, and write $p_{j_1}\coloneqq p^*$ and $x^*\coloneqq S(p^*)$ for convenience.
Then, properties (1)--(4) used in the previous case hold for the current $W$.
Indeed, (1) is clear by the definition of $W$, while (2)--(4) can be shown similarly as before using the fact that $S'$ cannot be improved by flipping along any subpath of $R$ (see the appendix for a more detailed argument). In addition, now $W$ is an $S$-alternating path. 
Let $\tilde{R}$ be the subpath of the reverse of $W$ from $q$ to $x^*=S(p^*)$. Observe that the last slot $x^*$, which is an intermediate party in $W$, prefers $S'(x^*)$ and $S(x^*)$ equally and $S(x^*)=p^*$ has a larger index than $q$ by property (2). Then, by a similar argument as before, properties (1)--(4) imply that $\tilde{R}$ is a blocking path for $S$ starting at $q$.
Since $q$ has an index smaller than that of $p^*$, this contradicts the minimality of the index of $p^*$.
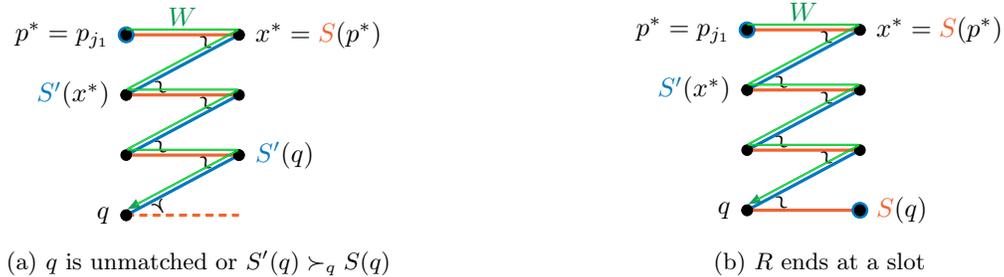
\begin{figure}[htbp]
\begin{minipage}[t]{.5\textwidth}
\centering
\begin{tikzpicture}[xscale=1.5,yscale=.8,very thick,font=\small]
\foreach \y [count=\i] in {4, 3, ..., 0}
{
    \coordinate(p\i) at (0,\y);
    \coordinate(s\i) at (1,\y);
}
\foreach \x / \y in {p1/s1, p2/s2, p3/s3}
{
    \draw[myred] (\x) -- (\y);
}
\draw[myred,densely dashed] (p4) -- (s4);
\foreach \x / \y in {s1/p2, s2/p3, s3/p4}
{
    \draw[myblue] (\x) -- (\y);
}
\foreach \a / \l in {p2/\sim, p3/\sim, p4/\succ}
{
    \node[font=\tiny,rotate=-60] at ($(\a)+(0.3,.13)$) {$\boldsymbol{\l}$};
}
\foreach \a / \l in {s1/\sim, s2/\sim, s3/\sim}
{
    \node[font=\tiny,rotate=-60] at ($(\a)+(-0.3,-.13)$) {$\boldsymbol{\l}$};
}
\foreach \a in {p1, p2, p3, p4, s1, s2, s3}
{
    \node[dot] at (\a) {};
}
\node[dot=6pt,fill=myblue,draw=myblue] at (p1) {};
\node[dot,label=left:{$p^*=p_{j_1}$}] at (p1) {};
\node[dot,label=left:{$\textcolor{myblue}{S'}(x^*)$}] at (p2) {};
\node[dot,label=left:{$q$}] at (p4) {};
\node[dot,label=right:{$x^*=\textcolor{myred}{S}(p^*)$}] at (s1) {};
\node[dot,label=right:{$\textcolor{myblue}{S'}(q)$}] at (s3) {};
\begin{scope}[transform canvas={yshift=2pt},transparency group,opacity=.9]
    \draw[path] (p1) -- (s1) -- (p2)  -- (s2) -- (p3) -- (s3) -- (p4);
    \draw[overpath] (p1) -- (s1) -- (p2)  -- (s2) -- (p3) -- (s3) -- (p4);
\end{scope}
\node[mygreen,font=\small] at (0.5,4.3) {$W$};
\end{tikzpicture}
\subcaption{$q$ is unmatched or $S'(q)\succ_q S(q)$}\label{subfig:case2normal}
\end{minipage}%
\begin{minipage}[t]{.5\textwidth}
\centering
\begin{tikzpicture}[xscale=1.5,yscale=.8,very thick,font=\small]
\foreach \y [count=\i] in {4, 3, ..., 0}
{
    \coordinate(p\i) at (0,\y);
    \coordinate(s\i) at (1,\y);
}
\foreach \x / \y in {p1/s1, p2/s2, p3/s3, p4/s4}
{
    \draw[myred] (\x) -- (\y);
}
\foreach \x / \y in {s1/p2, s2/p3, s3/p4}
{
    \draw[myblue] (\x) -- (\y);
}
\foreach \a / \l in {p2/\sim, p3/\sim, p4/\sim}
{
    \node[font=\tiny,rotate=-60] at ($(\a)+(0.3,.13)$) {$\boldsymbol{\l}$};
}
\foreach \a / \l in {s1/\sim, s2/\sim, s3/\sim}
{
    \node[font=\tiny,rotate=-60] at ($(\a)+(-0.3,-.13)$) {$\boldsymbol{\l}$};
}
\foreach \a in {p1, p2, p3, p4, s1, s2, s3, s4}
{
    \node[dot] at (\a) {};
}
\node[dot=6pt,fill=myblue,draw=myblue] at (p1) {};
\node[dot,label=left:{$p^*=p_{j_1}$}] at (p1) {};
\node[dot,label=left:{$\textcolor{myblue}{S'}(x^*)$}] at (p2) {};
\node[dot,label=left:{$q$}] at (p4) {};
\node[dot,label=right:{$x^*=\textcolor{myred}{S}(p^*)$}] at (s1) {};
\node[dot=6pt,fill=myblue,draw=myblue] at (s4) {};
\node[dot] at (s4) {};
\node[dot,label=right:{$\textcolor{myred}{S}(q)$}] at (s4) {};
\begin{scope}[transform canvas={yshift=2pt},transparency group,opacity=.9]
    \draw[path] (p1) -- (s1) -- (p2)  -- (s2) -- (p3) -- (s3) -- (p4);
    \draw[overpath] (p1) -- (s1) -- (p2)  -- (s2) -- (p3) -- (s3) -- (p4);
\end{scope}
\node[mygreen,font=\small] at (0.5,4.3) {$W$};
\end{tikzpicture}
\subcaption{$R$ ends at a slot}\label{subfig:case2end}
\end{minipage}%
\caption{Path $W$ in Case 2 of the proof of Proposition~\ref{prop:noblocking}}\label{fig:case2}
\end{figure}
\end{proof}

We now establish the group-strategyproofness of \Cref{alg:main} with augmented tie-breaking using Proposition~\ref{prop:noblocking}.
It is worth mentioning that even for strict preferences, existing proofs of strategyproofness rely either on the property that the outcome of the Gale--Shapley algorithm is independent of the order of proposals~\citep{DubinsFr81} or on the changes in the outcome when preferences are appropriately transformed~\citep{Roth82}. 
Our proof is more direct and does not rely on those properties, so we believe that it may be useful as an alternative proof for the case of strict preferences as well.

The following proposition is a rephrasing of \Cref{thm:SP}.

\begin{proposition}
\label{prop:SP-rephrase}
Let $I$ be any instance, $X$ be a non-empty subset of participants in $I$, and $I'$ be an instance obtained from $I$ by replacing the preference list $\succsim_p$ of each participant $p\in X$ with another (possibly identical) list $\succsim'_p$. 
Let $S$ and $S'$ be the matchings computed in the last iteration of the while-loop in \Cref{line:augment-eligibility} of \Cref{alg:main} with augmented tie-breaking applied to $I$ and $I'$, respectively. Then, there exists a participant $p\in X$ such that $p$ weakly prefers $S(p)$ to $S'(p)$ with respect to $\succsim_p$.
\end{proposition}
\begin{proof}
Suppose for contradiction that every manipulator $p\in X$ is better off, i.e.,  $S'(p)\succ_p S(p)$ for all $p\in X$.
Let us consider the \emph{last} iteration in the execution of the algorithm on the original instance~$I$ where a manipulator in $X$ was left unmatched; such an iteration exists since no participant is matched in the first iteration of the algorithm.
Let $S^*$ be the matching computed in that iteration. 

Consider an undirected bipartite graph $G$ whose vertex set is $P\cup L$ and edge set is $(S^*\setminus S')\cup (S'\setminus S^*)$ (where the edges are regarded as undirected edges). 
Note that all participants and slots are matched in $S'$ since it is the output matching for the instance $I'$.
Let $p^*\in X$ be a manipulator that is unmatched in $S^*$; such a manipulator exists by the definition of $S^*$. 
The connected component containing $p^*$ is a path starting at $p^*$ with an $S'$-edge and ending at some slot, say $x$, that is unmatched in $S^*$. 
Denote this path by $R$.  

Let $(a,b)$ be a pair of parties on $R$ satisfying the following conditions:
\begin{itemize}
\item $a$ is either the participant $p^*$ or an intermediate party that strictly prefers its succeeding party to its preceding party in the original preference $\succsim_a$.
\item $b$ is either the slot $x$ or an intermediate party that strictly prefers its preceding party to its succeeding party in the original preference $\succsim_b$.
\item $a$ and $b$ are distinct and appear on $R$ in this order, and every party $s$ between them prefers its preceding and succeeding parties equally in the original preference $\succsim_s$.
\end{itemize}
Such a pair $(a, b)$ must exist; it can be found by the following procedure. 
Initialize $(a,b)$ with $(p^*,x)$, and while there is a party between $a$ and $b$ that prefers its succeeding party to its preceding party (resp., its preceding party to its succeeding party), replace $a$ (resp., $b$) with that party.
Let $\hat{R}$ be the subpath of $R$ from $a$ to $b$.
We have four cases to consider: (a) $a$ is a participant and $b$ is a slot, (b) $a$ is a slot and $b$ is a participant, (c) $a$ and $b$ are both participants, and (d) $a$ and $b$ are both slots.

Note that for each manipulator $p\in X$, we have $S'(p)\succ_{p} S(p)$, which implies that the edge $(p, S'(p))$ is eligible when $S^*$ is computed.
This is because, in the iteration just before adding the team corresponding to $S(p)$ to $p$'s eligibility set in the execution of the algorithm on $I$, every team that $p$ strictly prefers to that team is eligible, and $p$ is unmatched. By definition, $S^*$ is computed in or after such an iteration, and hence $(p,S'(p))$ is eligible when $S^*$ is computed.
Additionally, if $p\in X$ is matched in $S^*$, then we have $S'(p)\succ_{p}S(p)\sim_{p}S^*(p)$; otherwise (i.e., if $S^*(p)\succ_p S(p)$), the participant $p$ must be left unmatched in an iteration later than the computation of $S^*$, which is a contradiction.
Since all participants on $\hat{R}$ except $a$ are matched in $S^*$ and weakly prefer their partners in $S^*$ to those in $S'$, we obtain that any participant on $\hat{R}$ except $a$ cannot belong to $X$.

\begin{figure}[t]
\begin{minipage}{.25\textwidth}
\centering
\begin{tikzpicture}[xscale=1.6,yscale=.8,very thick]
\draw[opacity=0] (-.5,0.5) grid[step=0.5] (2,5.5);
\node[dot,label=left:{$a$}] (p1) at (0,5) {};
\node[dot] (p2) at (0,4) {};
\node[dot] (p3) at (0,3) {};
\node[dot] (p4) at (0,2) {};
\node[dot=0] (p5) at (0,1) {};
\node[dot=0] (s1) at (1,5) {};
\node[dot] (s2) at (1,4) {};
\node[dot] (s3) at (1,3) {};
\node[dot] (s4) at (1,2) {};
\node[dot,label=right:{$b$}] (s5) at (1,1) {};
\draw[myred,densely dashed] (p1) -- (s1);
\draw[myred] (p2) -- (s2);
\draw[myred] (p3) -- (s3);
\draw[myred] (p4) -- (s4);
\draw[myred,densely dashed] (p5) -- (s5);
\draw[myblue] (p1) -- (s2);
\draw[myblue] (p2) -- (s3);
\draw[myblue] (p3) -- (s4);
\draw[myblue] (p4) -- (s5);
\node[font=\tiny,rotate=60] at (0.7,1.13) {$\boldsymbol{\prec}$};
\node[font=\tiny,rotate=60] at (0.7,2.13) {$\boldsymbol{\sim}$};
\node[font=\tiny,rotate=60] at (0.7,3.13) {$\boldsymbol{\sim}$};
\node[font=\tiny,rotate=60] at (0.7,4.13) {$\boldsymbol{\sim}$};
\node[font=\tiny,rotate=-120] at (0.3,1.87) {$\boldsymbol{\sim}$};
\node[font=\tiny,rotate=-120] at (0.3,2.87) {$\boldsymbol{\sim}$};
\node[font=\tiny,rotate=-120] at (0.3,3.87) {$\boldsymbol{\sim}$};
\node[font=\tiny,rotate=-120] at (0.3,4.87) {$\boldsymbol{\prec}$};
\begin{scope}[transform canvas={yshift=1.5pt},transparency group,opacity=.9]
\draw[path] (p1.center) -- (s2.center) -- (p2.center)  -- (s3.center) -- (p3.center) -- (s4.center) -- (p4.center) -- (s5.center);
\draw[overpath] (p1.center) -- (s2.center) -- (p2.center)  -- (s3.center) -- (p3.center) -- (s4.center) -- (p4.center) -- (s5.center);
\node[mygreen,font=\small] at (0.2,4.4) {$\hat{R}$};
\end{scope}
\end{tikzpicture}
\subcaption{$a\in P, b\in L$}\label{subfig:PL}
\end{minipage}%
\begin{minipage}{.25\textwidth}
\centering
\begin{tikzpicture}[xscale=1.6,yscale=.8,very thick]
\draw[opacity=0] (-.5,0.5) grid[step=0.5] (2,5.5);
\node[dot=0] (p1) at (0,5) {};
\node[dot] (p2) at (0,4) {};
\node[dot] (p3) at (0,3) {};
\node[dot,label=left:{$b$}] (p4) at (0,2) {};
\node[dot=0] (p5) at (0,1) {};
\node[dot=0] (s1) at (1,5) {};
\node[dot,label=right:{$a$}] (s2) at (1,4) {};
\node[dot] (s3) at (1,3) {};
\node[dot] (s4) at (1,2) {};
\node[dot=0] (s5) at (1,1) {};
\draw[myred] (p2) -- (s2);
\draw[myred] (p3) -- (s3);
\draw[myred] (p4) -- (s4);
\draw[myblue] (p1) -- (s2);
\draw[myblue] (p2) -- (s3);
\draw[myblue] (p3) -- (s4);
\draw[myblue] (p4) -- (s5);
\node[font=\tiny,rotate=60] at (0.7,2.13) {$\boldsymbol{\sim}$};
\node[font=\tiny,rotate=60] at (0.7,3.13) {$\boldsymbol{\sim}$};
\node[font=\tiny,rotate=60] at (0.7,4.13) {$\boldsymbol{\succ}$};
\node[font=\tiny,rotate=-120] at (0.3,1.87) {$\boldsymbol{\succ}$};
\node[font=\tiny,rotate=-120] at (0.3,2.87) {$\boldsymbol{\sim}$};
\node[font=\tiny,rotate=-120] at (0.3,3.87) {$\boldsymbol{\sim}$};
\begin{scope}[transform canvas={yshift=1.5pt},transparency group,opacity=.9]
\draw[path] (p4.center) -- (s4.center)  -- (p3.center) -- (s3.center) -- (p2.center) -- (s2.center);
\draw[overpath] (p4.center) -- (s4.center)  -- (p3.center) -- (s3.center) -- (p2.center) -- (s2.center);
\end{scope}
\node[mygreen,font=\small] at (-.1,4.4) {reverse of $\hat{R}$};
\end{tikzpicture}
\subcaption{$a\in L, b\in P$}\label{subfig:LP}
\end{minipage}%
\begin{minipage}{.25\textwidth}
\centering
\begin{tikzpicture}[xscale=1.6,yscale=.8,very thick]
\draw[opacity=0] (-.5,0.5) grid[step=0.5] (2,5.5);
\node[dot,label=left:{$a$}] (p1) at (0,5) {};
\node[dot] (p2) at (0,4) {};
\node[dot] (p3) at (0,3) {};
\node[dot,label=left:{$b$}] (p4) at (0,2) {};
\node[dot=0] (p5) at (0,1) {};
\node[dot=0] (s1) at (1,5) {};
\node[dot,label=right:{$\textcolor{myblue}{S'}(a)$}] (s2) at (1,4) {};
\node[dot] (s3) at (1,3) {};
\node[dot,label=right:{$\textcolor{myred}{S^*}(b)$}] (s4) at (1,2) {};
\node[dot=0] (s5) at (1,1) {};
\draw[myred,densely dashed] (p1) -- (s1);
\draw[myred] (p2) -- (s2);
\draw[myred] (p3) -- (s3);
\draw[myred] (p4) -- (s4);
\draw[myblue] (p1) -- (s2);
\draw[myblue] (p2) -- (s3);
\draw[myblue] (p3) -- (s4);
\draw[myblue] (p4) -- (s5);
\node[font=\tiny,rotate=60] at (0.7,2.13) {$\boldsymbol{\sim}$};
\node[font=\tiny,rotate=60] at (0.7,3.13) {$\boldsymbol{\sim}$};
\node[font=\tiny,rotate=60] at (0.7,4.13) {$\boldsymbol{\sim}$};
\node[font=\tiny,rotate=-120] at (0.3,1.87) {$\boldsymbol{\succ}$};
\node[font=\tiny,rotate=-120] at (0.3,2.87) {$\boldsymbol{\sim}$};
\node[font=\tiny,rotate=-120] at (0.3,3.87) {$\boldsymbol{\sim}$};
\node[font=\tiny,rotate=-120] at (0.3,4.87) {$\boldsymbol{\prec}$};
\begin{scope}[transform canvas={yshift=1.5pt},transparency group,opacity=.9]
\draw[path] (p1.center) -- (s2.center) -- (p2.center)  -- (s3.center) -- (p3.center) -- (s4.center) -- (p4.center);
\draw[overpath] (p1.center) -- (s2.center) -- (p2.center)  -- (s3.center) -- (p3.center) -- (s4.center) -- (p4.center);
\end{scope}
\node[mygreen,font=\small] at (0.2,4.4) {$\hat{R}$};
\end{tikzpicture}
\subcaption{$a,b\in P$}\label{subfig:PP}
\end{minipage}%
\begin{minipage}{.25\textwidth}
\centering
\begin{tikzpicture}[xscale=1.6,yscale=.8,very thick]
\draw[opacity=0] (-.5,0.5) grid[step=0.5] (2,5.5);
\node[dot=0] (p1) at (0,5) {};
\node[dot] (p2) at (0,4) {};
\node[dot] (p3) at (0,3) {};
\node[dot] (p4) at (0,2) {};
\node[dot=0] (p5) at (0,1) {};
\node[dot=0] (s1) at (1,5) {};
\node[dot,label=right:{$a$}] (s2) at (1,4) {};
\node[dot] (s3) at (1,3) {};
\node[dot] (s4) at (1,2) {};
\node[dot,label=right:{$b$}] (s5) at (1,1) {};
\draw[myred] (p2) -- (s2);
\draw[myred] (p3) -- (s3);
\draw[myred] (p4) -- (s4);
\draw[myred,densely dashed] (p5) -- (s5);
\draw[myblue] (p1) -- (s2);
\draw[myblue] (p2) -- (s3);
\draw[myblue] (p3) -- (s4);
\draw[myblue] (p4) -- (s5);
\node[font=\tiny,rotate=60] at (0.7,1.13) {$\boldsymbol{\prec}$};
\node[font=\tiny,rotate=60] at (0.7,2.13) {$\boldsymbol{\sim}$};
\node[font=\tiny,rotate=60] at (0.7,3.13) {$\boldsymbol{\sim}$};
\node[font=\tiny,rotate=60] at (0.7,4.13) {$\boldsymbol{\succ}$};
\node[font=\tiny,rotate=-120] at (0.3,1.87) {$\boldsymbol{\sim}$};
\node[font=\tiny,rotate=-120] at (0.3,2.87) {$\boldsymbol{\sim}$};
\node[font=\tiny,rotate=-120] at (0.3,3.87) {$\boldsymbol{\sim}$};
\begin{scope}[transform canvas={yshift=1.5pt},transparency group,opacity=.9]
\draw[path] (s2.center) -- (p2.center)  -- (s3.center) -- (p3.center) -- (s4.center) -- (p4.center) -- (s5.center);
\draw[overpath] (s2.center) -- (p2.center)  -- (s3.center) -- (p3.center) -- (s4.center) -- (p4.center) -- (s5.center);
\end{scope}
\node[mygreen,font=\small] at (0.2,4.4) {$\hat{R}$};
\end{tikzpicture}
\subcaption{$a,b\in L$}\label{subfig:LL}
\end{minipage}%
\caption{The four cases in the proof of Proposition~\ref{prop:SP-rephrase}.
The preferences shown in the figures represent the ones in the original instance $I$.}
\label{fig:path-ab}
\end{figure}
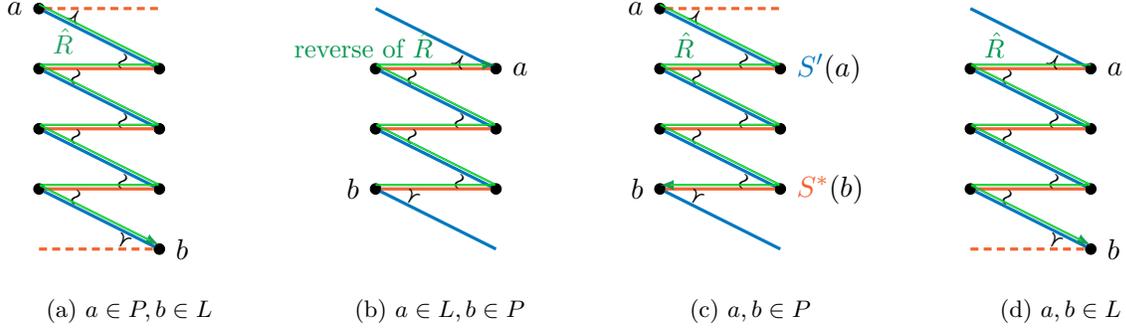

By Proposition~\ref{prop:noblocking}, we have that $S^*$ has no blocking path with respect to the eligibility sets when $S^*$ is computed.
Proposition~\ref{prop:noblocking} also implies that $S'$ has no blocking path in the instance $I'$ with respect to the eligibility sets when $S'$ is computed. 
We consider four cases depending on whether each of $a$ and $b$ is a participant or a slot.

\smallskip
\noindent\textit{Case (a). }
If $a$ is a participant and $b$ is a slot, then the path $\hat{R}$ forms a blocking path for $S^*$ within the eligibility sets when $S^*$ is computed, a contradiction (see \Cref{subfig:PL}).

\smallskip
\noindent\textit{Case (b). }
If $a$ is a slot and $b$ is a participant, consider the reverse of the path $\hat{R}$ (see \Cref{subfig:LP}). 
This is a path from the participant $b$ to the slot $a$ and does not contain manipulators in $X$, and hence the preference lists of all parties on this path are the same in $I$ and $I'$. Then, in the instance $I'$, this path forms a blocking path for $S'$ within the eligibility sets when $S'$ is computed, a contradiction.

\smallskip
\noindent\textit{Case (c). }
We next assume that $a$ and $b$ are both participants (see \Cref{subfig:PP}). In this case, if the index of $b$ is larger than that of $a$, then the subpath of $\hat{R}$ from $a$ to $S^*(b)$ (i.e., the slot preceding~$b$) forms a blocking path for $S^*$, a contradiction. Thus, $a$'s index is larger than $b$'s index.
Now, consider the reverse of the subpath of $\hat{R}$ from $S'(a)$ to $b$, denoted by $\tilde{R}$. 
Then, $\tilde{R}$ is a path from the participant $b$ to the slot $S'(a)$ and does not contain manipulators in $X$ since $a$ is not on the path. 
Since the preference lists of all parties on $\tilde{R}$ are the same in $I$ and $I'$, $b$ strictly prefers the succeeding slot on $\tilde{R}$ to $S'(b)$, and all other parties on $\tilde{R}$ prefer their partners in $S^*$ and $S'$ equally. 
As the index of $S'(S'(a))=a$ is larger than that of $b$, it follows that $\tilde{R}$ forms a blocking path for~$S'$ in the instance $I'$, a contradiction.   

\smallskip
\noindent\textit{Case (d). }
Finally, assume that $a$ and $b$ are both slots (see \Cref{subfig:LL}). Then, $\hat{R}$ is an even-length path using edges in $S^*$ and $S'$ alternately and does not contain manipulators in $X$. 
Note that all participants $\hat{p}$ on $\hat{R}$ prefer $S'(\hat{p})$ and $S^*(\hat{p})$ equally, and hence all the pairs on $\hat{R}$ are eligible when $S'$ is computed in the execution of the algorithm on $I'$ and when $S^*$ is computed in the execution on~$I$. 
Consider modifying the matching $S^*$ by flipping along $\hat{R}$ and, in addition, removing the pair in $S^*$ that covers $b$ if such a pair exists. In the resultant matching, the slot $a$ is worse off while $b$ is better off compared to $S^*$. Since $S^*$ was chosen over this matching, the slot $a$ must have a higher priority than the slot $b$.
Next, consider modifying the matching $S'$ by flipping along $\hat{R}$ and removing the pair in $S'$ that covers $a$. In the resultant matching, the slot $b$ is worse off while $a$ is better off compared to $S'$. Since $S'$ was chosen over this matching, the slot $b$ must have a higher priority than the slot $a$, which contradicts the previous conclusion.
\end{proof}

We remark that \Cref{alg:main} is not (individually) strategyproof for teams.
Indeed, a classic result of \citet[Thm.~3]{Roth82} implies that even in one-to-one matching, no procedure can satisfy participant-justified EF along with strategyproofness for both sides.

\section{Discussion}
\label{sec:discussion}

In this paper, we have presented an algorithm for many-to-one matching (\Cref{alg:main}), which unifies both the round-robin algorithm from fair division and the Gale--Shapley algorithm from two-sided matching.
Not only does our algorithm provide fairness to both sides in the potential presence of indifferences, but it also satisfies several attractive properties including Pareto optimality, balancedness, and group-strategyproofness on the participant side.
Moreover, the algorithm can be implemented in polynomial time.
We believe that fairness among teams brings a new perspective that has been largely ignored in the extensive literature of two-sided matching thus far, and this perspective has a potential to be valuable in several applications.

It is worth noting that our results also provide new insights in the context of stable matching with indifferences~\citep{ErdilEr08,Kamiyama14,ErdilEr17,DomanicLaPl17,ErdilKu19,ErdilKiKu22,BandoImKa25}. 
In particular, consider the following extension of our setting.
\begin{enumerate}
\item Each team $i \in T$ has a quota $q_i$. Without loss of generality, we assume that $q_i\le m$.
\item Each participant $p \in P$ may prefer being unassigned to being assigned to certain teams, i.e., each participant has a weak transitive preference $\succsim_p$ over $T \cup \{\varnothing\}$.
\item Each team $i \in T$ may prefer not to be assigned certain participants, i.e., each team has a weak transitive preference $\succsim_i$ over $P \cup \{\varnothing\}$.
\end{enumerate}
In this extended setting, an allocation is allowed to be partial, i.e., it may leave some participants unassigned.
A (possibly partial) allocation $A$ is called \emph{(weakly) stable} if 
\begin{enumerate}[(i)]
\item it satisfies participant-justified EF;
\item $p\succsim_i\varnothing$ and $i\succsim_p\varnothing$ for every team $i$ and participant $p$ such that $p\in A_i$ (individually rational);
\item there do not exist $p\in A_i$ and $i'\in T$ such that $i'\succ_p i$ and $|A_{i'}|<q_{i'}$ (non-wasteful).
\end{enumerate}

Now, consider applying \Cref{alg:main} together with the additional tie-breaking rule described in \Cref{sec:SP}.
We construct an augmented instance as follows.
Let $z=\sum_{i\in T}q_i$ be the sum of quotas across all teams.
We create $z$ dummy participants representing ``unassigned'', each of whom values every team (including the dummy team that we will create) equally.
Similarly, we create a dummy team representing ``unassigned'' with quota $m$; this team values all participants (including the dummy participants) equally. For each team (including dummy), we create as many slots as its quota.
In the preference list of each real participant, $\varnothing$ is replaced by the dummy team; likewise, in the preference list of each real team, $\varnothing$ is replaced by dummy participants, all of whom are tied.

Observe that the augmented instance has $z+m$ team-slots and $z+m$ participants.
Moreover, there are sufficiently many dummy team-slots that every real participant can be assigned to a dummy team-slot rather than a real team-slot if she wishes; a similar statement holds for dummy participants and real teams.
To this instance, we apply \Cref{alg:main} with augmented tie-breaking, using the team-slots we create above instead of those from Lines 1--4. 
The number of slots may differ across teams in this instance, but priorities over slots (used in Line 12) are again defined in round-robin order, where teams with fewer slots are skipped in later rounds. 
In the output, the dummy team and participants are removed.
With these modifications, the algorithm simultaneously satisfies stability, PO, and group-strategyproofness for participants due to the corresponding properties of \Cref{alg:main}.
We remark that \citet{DomanicLaPl17} also established a mechanism with these properties by reducing the problem to the ``generalized assignment game''~\citep{DemangeGa85}. In contrast, our algorithm provides a more direct approach without such a reduction.

We conclude the paper by briefly discussing fairness among teams for the aforementioned extension.
With different quotas, the resulting allocation may no longer be balanced.
The definition of team-justified SD-EF1 needs to be adjusted to take into account the quotas as well as the preferences of being unassigned.

\begin{definition}
An allocation $A$ satisfies \emph{extended team-justified SD-EF1} if the following holds: For all distinct $i,j\in T$ where team~$i$ has quota~$z_i$, let $B_j$ denote the set of participants in $A_j$ who weakly prefer team~$i$ to team~$j$ and are strictly better than unassigned for team~$i$.
Then, for every subset $C_j\subseteq B_j$ of size at most $z_i$, it holds that $A_i\succsim_i^{\text{SD}} (C_j\setminus X)$ for some $X\subseteq C_j$ with $|X|\le 1$.
\end{definition}

Using an argument similar to that in the proof of \Cref{thm:main-properties}, one can show that the output of the augmented algorithm satisfies extended team-justified SD-EF1.
\begin{theorem}
There exists a polynomial-time algorithm for the extended setting that returns an allocation satisfying stability, extended team-justified SD-EF1, and PO, and is moreover group-strategyproof for participants.
\end{theorem}

\section*{Acknowledgments}

This work was partially supported by JSPS KAKENHI Grant Numbers JP17K12646, JP20K19739, JP21K17708, JP21H03397, and JP25K00137, by JST PRESTO Grant Numbers JPMJPR2122, JPMJPR20C1, and JPMJPR212B, by JST ERATO Grant Number JPMJER2301, by JST CRONOS Japan Grant Number JPMJCS24K2, by Value Exchange Engineering, a joint research project between Mercari, Inc.\ and the RIISE, by the Singapore Ministry of Education under grant number MOE-T2EP20221-0001, and by an NUS Start-up Grant.
We would like to thank the WINE 2025 reviewers for their valuable comments.

\bibliographystyle{plainnat}
\bibliography{main}

\appendix

\section{Omitted Details in the Proof of Proposition~\ref{prop:noblocking}}

In the analysis of Case 2 in the proof of Proposition~\ref{prop:noblocking}, we claimed that the walk $W$ satisfies properties (1)--(4), which are presented below again.
\begin{enumerate}[(1)]
\item $W$ is a path that starts at a participant unmatched in $S'$ and uses pairs not in $S'$ and those in~$S'$ alternately.
\item Every participant on $W$ other than $p_{j_1}=p^*$ has an index smaller than $j_1$. In particular, $q$ does.
\item Every intermediate party on $W$ prefers its preceding and succeeding parties equally.
In addition, $x^*$ prefers the preceding participant, $S'(x^*)$, and $S(x^*)$ all equally.
\item The last participant $q$ is either unmatched in $S$ or strictly prefers $S'(q)$ to $S(q)$.
\end{enumerate}
Here we provide a detailed proof of this claim.
First, recall that in Case 2, the walk $W$ is defined as a subpath $R[p^*, q]$ of path $R$, which starts at the participant $p_{j_1}=p^*$ that is unmatched in $S'$, and uses edges in $S\setminus S'$ and $S'\setminus S$ alternately. 
This implies property (1) immediately. 
We next show properties (2)--(4).
Recall that, by the choice of $q$, any participant $p$ on $W$ other than $q$ satisfies $S(p)\succsim_p S'(p)$. 
In addition, if $q$ is matched in $S$ and satisfies $S(q)\succsim_q S'(q)$, then $R$ ends at the slot succeeding $q$ that is unmatched in $S'$.
Recall also that we set $p_{j_1}\coloneqq p^*$ and $x^*\coloneqq S(p^*)$ for notational convenience, and that $(p^*, x^*)\in S$ is eligible when $S'$ is computed.

\smallskip
\noindent\textit{Properties (2)--(3). }
Since $W$ starts with an $S$-edge and uses edges in $S$ and $S'$ alternately, for any slot $x$ on $W$, its preceding and succeeding participants are $S(x)$ and $S'(x)$, respectively. As $S'$ is computed earlier than $S$, Lemma~\ref{lem:successively-better}  implies that $S(x)\succsim_x S'(x)$.

As for the first slot $x^*= S(p^*)$ on $W$, we see that $S'(x^*) \succsim_{x^*} S(x^*)$; otherwise, we could improve~$S'$ by flipping along $W[p^*, S'(x^*)]$, i.e., adding $(p^*, x^*)$ (which is the same as $(S(x^*),x^*)$) and removing $(S'(x^*), x^*)$, which would make $x^*$ better off. 
It then follows that $S(x^*)\sim_{x^*}S'(x^*)$. As $S(x^*)$ is the participant preceding $x^*$ in $W$, property (3) holds for $x^*$. 

For the remaining parties on $W$, we show properties (2) and (3) by induction.
Take a participant~$p$ on $W$ other than $p_{j_1}=p^*$, and suppose that property (3) holds for all parties on $W$ preceding $p$. 
The slot preceding $p$ in $W$ is $S'(p)$ and the one succeeding $p$ is $S(p)$ if it exists. 
The participant $p$ must have an index smaller than that of $p^*$; otherwise, we could improve $S'$ by flipping along $W[p^*,p]$. 
Since $p$ has an index smaller than that of $p^*=p_{j_1}$, by the choice of $p^*$, we have that $p$ satisfies condition (b), i.e., $p$ is matched to a slot of a least-preferred eligible team in $S$. 
Since the eligibility set is monotone increasing and $S'$ is computed earlier than $S$, this implies that $S'(p)\succsim_p S(p)$ whenever $p$ is matched in $S$. 
If $p\neq q$, then by the choice of $q$, the participant $p$ is matched in $S$ and $S'(p)\not\succ_p S(p)$, and hence we have $S'(p)\sim_p S(p)$. 
Thus, properties (2) and~(3) hold for $p$. 
In addition, when $p\neq q$, the slot $y\coloneqq S(p)$, which succeeds $p$ in $W$, is matched in~$S'$ and satisfies $S'(y)\succsim_y S(y)=p$; otherwise, we could improve $S'$ by flipping along $W[p^*, S'(y)]$, which would make $y$ better off. 
At the same time, Lemma~\ref{lem:successively-better} implies that $S(y)\succsim_y S'(y)$. 
Thus, we have  $S'(y)\sim_{y}S(y)$, i.e., property (3) holds for the slot succeeding $p$ in $W$.

\smallskip
\noindent\textit{Property (4). }
Finally, we show that $q$ is either unmatched in $S$ or satisfies $S'(q)\succ_{q} S(q)$.
Suppose for contradiction that this does not hold. 
This means that $q$ is matched in $S$ and $S(q)\succsim_q S'(q)$.
By the argument in the previous paragraph, we also have $S'(q)\succsim_q S(q)$, so it must be that $S'(q)\sim_q S(q)$.
Moreover, by the definition of $q$, the path $R$ ends at the slot $S(q)$ which is unmatched in $S'$. 

Then, $W+(q,S(q))=R$ holds and it is an $S'$-alternating path whose first and last parties are unmatched in $S'$ and all intermediate parties prefer the preceding and succeeding parties equally. Hence, we can improve $S'$ by flipping along $R$, a contradiction.
\smallskip

It follows that the walk $W$ satisfies properties (1)--(4) in Case 2.

\end{document}